\documentclass[a4paper,10pt]{article}

\usepackage{authblk}
\usepackage{hyperref}
\usepackage{amsmath,amsfonts}
\usepackage{amssymb}
\usepackage{amsthm}
\usepackage{a4wide}
\usepackage{color}
\usepackage[pdftex]{graphicx}
\usepackage{epstopdf}

\theoremstyle{definition}\newtheorem{dfn}{Definition}[section]
\theoremstyle{plain}\newtheorem{obs}{Observation}[section]
\theoremstyle{plain}\newtheorem{thm}{Theorem}[section]
\theoremstyle{plain}\newtheorem{lem}{Lemma}[section]
\theoremstyle{plain}\newtheorem{cor}{Corollary}[section]
\theoremstyle{plain}\newtheorem{claim}{Claim}[section]

\newcommand{\steven}{\textcolor{black}}
\newcommand{\twu}{\textcolor{black}}

\title{\textbf{Treewidth distance on phylogenetic trees}}
\author[1]{Steven Kelk \thanks{steven.kelk@maastrichtuniversity.nl}}
\author[1]{Georgios Stamoulis \thanks{georgios.stamoulis@maastrichtuniversity.nl}}
\author[2]{Taoyang Wu \thanks{taoyang.wu@uea.ac.uk}}
\affil[1]{Department of Data Science and Knowledge Engineering (DKE), Maastricht University, Maastricht, The Netherlands}
\affil[2]{School of Computing Sciences, University of East Anglia, Norwich, United Kingdom}

\date{}

\begin{document}
\maketitle
\begin{abstract}
\noindent In this article we study the treewidth of the \emph{display graph}, an auxiliary graph structure obtained from the fusion of phylogenetic (i.e., evolutionary) trees
at their leaves. Earlier work has shown that the treewidth of the display graph is bounded if the trees are in some formal sense topologically similar. Here we further expand
upon this relationship. We analyse a number of reduction rules which are commonly used in the phylogenetics  literature to obtain fixed parameter tractable algorithms. In
some cases (the \emph{subtree} reduction) the reduction rules behave similarly with respect to treewidth, while others (the \emph{cluster} reduction) behave very differently,
and the behaviour of the \emph{chain reduction} is particularly intriguing because of its link with graph separators and forbidden minors. We also show that the gap between
treewidth and Tree Bisection and Reconnect (TBR) distance can be infinitely large, and that unlike, for example, planar graphs the treewidth of the display graph can be as
much as linear in its number of vertices. On a slightly different note we show that if a display graph is formed from the fusion of a phylogenetic network and a tree, rather
than from two trees, the treewidth of the display graph is bounded whenever the tree can be topologically embedded (``displayed'') within the network. This opens the door to
the formulation of the display problem in Monadic Second Order Logic (MSOL). A number of other auxiliary results are given. We conclude with a discussion and list a number of
open problems.
\end{abstract}


\section{Introduction}

Phylogenetic trees are used extensively within computational biology to model the history of a set of species $X$; the internal nodes represent evolutionary diversification
events such as speciation \cite{SempleSteel2003}. Within the field of phylogenetics there has long been interest in quantifying the topological dissimilarity of phylogenetic
trees and understanding whether this dissimilarity is biologically significant. This has led to the development of many \emph{incongruency measures} such as Subtree Prune and
Regraft (SPR) distance and Tree Bisection and Reconnect (TBR) distance \cite{AllenSteel2001}. These measures are often \textbf{NP}-hard to compute. More recently such measures
have also attracted attention because of their importance in methods which merge dissimilar trees into \emph{phylogenetic networks}; phylogenetic networks are the
generalization of trees to graphs \cite{HusonRuppScornavacca10}.

Parallel to such developments there has been growing interest in the role of the graph-theoretic parameter \emph{treewidth} within phylogenetics. Treewidth is an intensely
studied parameter in algorithmic graph theory and it indicates, at least in an algorithmic sense, how far an undirected graph is from being a tree (see e.g.
\cite{bodlaender1994tourist,bodlaender2010treewidth, bodlaender2011treewidth} for background). The enormous focus on treewidth is closely linked to the fact that a great many
\textbf{NP}-hard optimization problems become (fixed parameter) tractable on graphs of bounded treewidth\cite{Cygan:2015:PA:2815661}. A seminal paper by Bryant and Lagergren
\cite{bryant2006compatibility} linked phylogenetics to treewidth by demonstrating that, if a set of trees (not necessarily all on the same set of taxa $X$) can simultaneously
be topologically embedded within a single ``supertree'' - a property known as \emph{compatibility} -  then an auxiliary graph known as the \emph{display graph} has bounded
treewidth. Since this paper a small but growing number of papers at the interface of graph theory and phylogenetics have explored this relationship further. Much of this
literature focuses on the link between compatibility and (restricted) triangulations of the display graph (e.g. \cite{vakati2011graph,gysel2012reducing,Vakati2015337}), but
more recently the algorithmic dimension has also been tentatively explored \cite{baste2016efficient,grigoriev2015,kelk2015reduction}. In the spirit of the original Bryant and
Lagergren paper, which used heavy meta-theoretic machinery to derive a theoretically efficient algorithm for the compatibility problem, Kelk et al \cite{kelk2015} showed that
the treewidth of the display graph of two trees is linearly bounded as a function of the TBR distance (equivalently, the size of a Maximum Agreement Forest - MAF
\cite{AllenSteel2001}) between the two trees, and then used this insight to derive theoretically efficient algorithms for computation of many different incongruency measures.
In this article it was empirically observed that in practice the treewidth of the display graph is often much smaller than the TBR distance (and thus also of the many
incongruency measures for which TBR is a lower bound). This raised two natural questions. First, in how far can this apparently low treewidth be exploited to yield genuinely
practical dynamic programming algorithms running over low-width tree decompositions? There has been some progress in this direction in the compatibility literature (notably,
\cite{baste2016efficient}) but there is still much work to be done. Second, how \emph{exactly} does the treewidth of the display graph behave, both in the sense of extremal
results (e.g. how large can the treewidth of a display graph get?) and in the sense of understanding when and why the treewidth differs significantly from measures such as TBR. 

Here we focus primarily on the second question.  We start by analyzing how reduction rules often used in the computation of incongruency measures impact upon the treewidth of
the display graph. Not entirely surprisingly the \emph{common pendant subtree} reduction rule \cite{AllenSteel2001} is shown to preserve treewidth. The \emph{cluster} reduction
\cite{BSS06,linz2011cluster,bordewich2017fixed}, however, behaves very differently for treewidth than for many other incongruency measures. Informally speaking, if both trees
can be split by deletion of an edge into two subtrees on $X'$ and $X''$, many incongruency measures combine additively around this \emph{common split}, while treewidth behaves
(up to additive terms) like the maximum function. We use this later in the article to explicitly construct a family of tree pairs such that the treewidth of the display graph
is 3, but the TBR distance of the trees (and their MP distance - a measure based on the phylogenetic principle of parsimony
\cite{fischer2014maximum,moulton2015parsimony,kelk2015reduction}) grows to infinity. The third reduction rule we consider is the \emph{chain rule}, which collapses common
caterpillar-like regions of the trees into shorter structures. For incongruence measures it is often the case that truncation of such chains to $O(1)$ length preserves the
measure \cite{AllenSteel2001,sempbordfpt2007,whidden2015calculating}, although sometimes the weaker result of truncation to length $f(k)$
\cite{ierselLinz2013,vanIersel20161075} (for some function that depends only on the incongruency parameter $k$) is the best known. We show that truncation of common chains to
length $f(tw)$, where $tw$ is the treewidth of the display graph, indeed preserves treewidth; this uses asymptotic results on the number of vertices and edges in forbidden
minors for treewidth. Proving that truncation to $O(1)$-length preserves treewidth remains elusive; we prove the intermediate result that truncation to length \steven{2} can
cause the treewidth to decrease by at most 1. The case when the chain is not a separator of the display graph seems to be a particularly challenging bottleneck in removing the 	
``$-1$'' term from this result. Although intuitively reasonable, it remains unclear whether truncation to length $O(1)$ is treewidth-preserving, for any universal constant.

In the next section we adopt a more structural perspective. We show that, given an arbitrary (multi)graph $G$ on $n$ vertices with maximum degree $k$, one can construct two
unrooted binary trees $T_1(G)$ and $T_2(G)$ such that their display graph $D = D(T_1(G), T_2(G))$ has at most $O(nk)$ vertices and edges and $G$ is a minor of $D$. We combine
this with the known fact that cubic expanders (a special family of 3-regular graphs) on $n$ vertices have treewidth $\Omega(n)$ to yield the result that display graphs on $n$
vertices can also (in the worst case) have treewidth linear in $n$. This contrasts, for example, with planar graphs on $n$ vertices which have treewidth at most $O(\sqrt{n})$
\cite{diestel2010}. We also show how a more specialized construction can be used to embed arbitrary grid minors \cite{chuzhoy2015excluded} into display graphs with a much
smaller inflation in the number of vertices and edges

Moving on, we then switch to the topic of phylogenetic \emph{networks}. In our context, networks can be considered to be connected undirected graphs whose internal nodes all
have degree 3 and whose leaves are bijectively labelled by a set of labels $X$ \cite{van2016unrooted,GBP2012}. Due to the interpretation of phylogenetic networks as species
networks that contain multiple embedded gene trees, a major algorithmic question in that field is to determine whether a network \emph{displays} (i.e. topologically embeds) a
tree \cite{ISS2010b}. Here we show that constructing a display graph from a network and a tree (rather than two trees) also has potential applications. Specifically, we show
that if $N$ displays $T$ the display graph of $N$ and $T$ has bounded treewidth (as a function of the treewidth of $N$ or, alternatively,  as a a function of the
\emph{reticulation number} of $N$). This then allows us to pose the question ``does $N$ display $T$?'' in Monadic Second Order Logic (MSOL) which yields a logical-declarative
version of earlier, combinatorial fixed parameter tractability results. This once again shows that the flexibility of Courcelle's Theorem \cite{Courcelle90,Arnborg91} in the
context of phylogenetics; the details of the MSOL formulation are technical and are deferred to the appendix. For completeness we show that treewidth alone is not sufficient to
distinguish between YES and NO instances of the display problem: we show a network $N$ and a tree $T$ such that $N$ does not display $T$, but the treewidth of $N$ does not
increase when merged into a display graph with $T$.

In the last two mathematical sections of the paper we prove that, if two trees have TBR distance 1, or MP-distance 1, then the treewidth of their display graph is 3. However,
the converse certainly does not hold: we construct the aforementioned ``infinite gap'' examples where the display graph has treewidth 3 but both TBR distance and MP-distance
spiral off to infinity.

Finally, we reflect on the wider context of these results and discuss a number of open problems.

In conclusion, we observe that for (algorithmic) graph theorists the interface between treewidth and phylogenetics continues to yield many new questions which will likely
require a new ``phylo-algorithmic'' graph theory to be answered.  For phylogeneticists the appeal remains structural-algorithmic: can we convert the apparently low treewidth of
display graphs into competitive, or even superior, algorithms for computation of incongruency measures?

\section{Preliminaries}
An \textit{unrooted binary phylogenetic tree} $T$ on a set of leaf labels (known as \textit{taxa}) $X$ is an undirected tree where all internal vertices have degree three and the
leaves are bijectively labeled by $X$. If we \steven{(exceptionally)} allow some internal vertices of $T$ to have degree two, then we call these vertices \textit{roots} (abusing slightly the usual root
meaning). Similarly, an \textit{unrooted binary phylogenetic network} $N$ on a set of leaf labels $X$ is a simple, connected, undirected graph that has $|X|$ degree-1 vertices that are
bijectively labeled by $X$ and any other vertex has degree 3. The \textit{reticulation number} $r(N)$ of $N$ is defined as $r(N) := |E| - (|V|-1)$, i.e., the number of edges we
need to delete from $N$ in order to obtain a \steven{tree that spans $V$}. A network $N$ with $r(N) = 0$ is simply an unrooted phylogenetic tree. \steven{When it is understood from the context we will often drop the prefix ``unrooted binary phylogenetic'' for brevity}.


Let $Y \subseteq X$. Then, for a tree $T$ we denote by \steven{$T|Y$} the tree which is obtained by forming a minimal subgraph $T'$ of $T$ that spans all
leaves labeled by $Y$, and suppressing any vertices of degree 2.

Let $T_1, T_2$ be two trees (or networks, or combination of the trees and networks), both on the same set of leaf labels $X$. The \emph{display graph} of
$T_1, T_2$, denoted by $D(T_1,T_2)$, is formed by identifying vertices with the same leaf label and forming the disjoint union of these two trees. This can be extended in a
straightforward way to more than 2 trees/networks.

If $N$ is a network and $T$  a tree on a common set of taxa $X$ we say that $N$ \textit{displays} $T$ if there exists a \textit{subtree} $N'$ of $N$ that is a
subdivision of $T$. In other words, $T$ can be obtained by a series of edge contractions on a subgraph $N'$ of $N$. $N'$ is \steven{a minimal} connected subgraph of $N$ that spans all
the taxa $X$.  We say that $N'$ is an \emph{image} of $T$.  We can easily see that every vertex
\steven{of $T$ is mapped to a vertex of $N'$, and that edges of $T$ potentially map to
paths in $N'$},
leading us to the following observation (see also \cite{bryant2006compatibility}):


\begin{obs}
If an unrooted binary network $N$ displays an unrooted binary tree $T$, where both are defined on a common set of leaves (taxa) $X$, then there exists a surjection $f$ from a subtree $N'$ of $N$ to
$T$ such that: (1) $f(\ell) = \ell, \forall \ell \in X$, (2)  \steven{the subsets of $V(N')$ induced by $f^{-1}$ are mutually disjoint, and each such subset induces a connected subtree of $V(N')$,}
and (3) $\forall \{ u,v \} \in E(T), \exists_1 \{ \alpha, \beta \} \in E(N'):$ $f(\alpha) = u$ and $f(\beta) = v$.
\end{obs}

This observation will be crucial when we study the treewidth of $D(N,T)$ as a function of the treewidth of $N$. We say that two (or more) trees are \textit{compatible} if there
exists another tree on $X$ that displays all the trees. \steven{Note that for two unrooted
binary phylogenetic trees on the same set of labels $X$ compatibility is simply equivalent to
the existence of a label-preserving isomorphism between the two trees.}

A \emph{tree decomposition} of an undirected graph $G=(V,E)$ is a pair $(\mathcal{B}, \mathbb{T})$ where $\mathcal{B} = \{B_1, \dots ,B_q\}$, $B_i \subseteq V(G)$, is a collection of bags
and $\mathbb{T}$ is a tree whose nodes are the bags $B_i$ satisfying the following three properties:
\begin{itemize}
\item[(tw1)] $\cup_{i=1}^q B_i = V(G)$;
\item[(tw2)] $\forall e = \{ u,v \} \in E(G), \exists B_i \in \mathcal{B} \mbox{ s.t. } \{u,v\} \subseteq B_i$;
\item[(tw3)] $\forall v \in V(G)$ all the bags $B_i$ that contain $v$ form a connected subtree of $\mathbb{T}$.
\end{itemize}

The \emph{width} of $(\mathcal{B}, \mathbb{T})$ is equal to $\max_{i=1}^q |B_i|-1$. The \emph{treewidth} of $G$ is the smallest width among all possible tree decompositions of
$G$. For a graph $G$, we denote $tw(G)$ the treewidth of $G$. Given a tree decomposition $\mathbb{T}$ for some graph $G$, we denote by $V(\mathbb{T})$ the set of its bags and
by $E(\mathbb{T})$ the set of its edges (connecting bags). Property (tw3) is also known as \emph{running intersection property}. We note that the treewidth of any graph $G$ is
at most $|V(G)|-1$: consider a bag with all vertices of $G$. This is a valid tree decomposition of width $|V(G)|-1$. Thus the treewidth is always a finite parameter for any
graph.

Another, equivalent, definition of treewidth is based on chordal graphs. We remind that a graph $G$ is chordal if every induced cycle in $G$ has exactly three vertices. The
treewidth of $G$ is the minimum, ranging over \emph{all} chordal completions $c(G)$ of $G$ (we add edges until $G$ becomes a chordal graph), of the  size of the maximum clique
in $c(G)$ minus one. Under this definition, each bag of a tree decomposition of $G$ naturally corresponds to a maximal clique in the chordal completion of $G$ \cite{Blair1993}.

For a graph $G=(V,E)$ and an edge $e = \{u,v\} \in E(G)$, the \emph{deletion} of $e$ is the operation which simply deletes $e$ from $E(G)$ and leaves the rest of the graph $G$
the same. The \emph{contraction} of $e$, denoted $G/e$, is the operation where edge $e$ is deleted and its incident vertices $u,v$ are identified.
We say that a graph $H$ is a \emph{minor} of another graph $G$ if $H$ can be obtained by repeated applications of edge deletions and/or edge contraction, followed possibly by
deleting isolated vertices, on $G$\footnote{\steven{Equivalently we can say that $H$ is a minor of $G$ if $H$ can be obtained by vertex deletions, edge deletions and edge
contractions in $G$.}}.  The order that these operations are performed does not matter and it will result always $H$ in any performed order.


\subsection{Phylogenetic distances and measures}

Several distances have been proposed to measure the incongruence between two (or more) phylogenetic trees on the same set of taxa. The high-level problem is to propose a
measure that quantifies the \emph{dissimilarity} of a given set of phylogenetic trees.

The most relevant distances for the purpose of this article are the so-called \emph{Tree Bisection and Reconnect} distance and the \emph{Maximum Parsimony Distance} which are defined in
the following.

Given an unrooted binary phylogenetic tree $T$ on $X$, a \emph{Tree Bisection and Reconnect} (TBR) move is defined as follows \cite{AllenSteel2001}: (1) we delete an edge of
$T$ to obtain two subtrees $T'$ and $T''$. (2) Then we select two  edges $e_1 \in T', e_2 \in T''$, subdivide them with two new vertices $v_1$ and $v_2$ respectively, add an
edge from $v_1$ to $v_2$, and suppress all vertices of degree 2. In case either $T'$ or $T''$ is a single leaf, then the new edge connecting $T'$ and $T''$ is incident to that
leaf. Let $T_1, T_2$ be two unrooted binary phylogenetic trees on the same set of leaf-labels. The TBR-distance from $T_1$ to $T_2$, denoted $d_{TBR}(T_1, T_2)$, is the
\textit{minimum} number of TBR moves required to transform $T_1$ into $T_2$ (or, equivalently, $T_2$ to $T_1$). 

Computing the TBR-distance is essentially equivalent to the \textit{Maximum Agreement Forest (MAF)} problem: Given an unrooted binary phylogenetic tree on $X$ and $X' \subset X$ we let
$T(X')$ denote the minimal subtree that connects all the elements in $X'$. An \emph{agreement forest} of two unrooted binary trees $T_1, T_2$ on $X$ is a partition of $X$ into
non-empty blocks $\{X_1, \ldots, X_k\}$ such that (1) for each $i \neq j$, $T_1(X_i)$ and $T_1(X_j)$ are node-disjoint and $T_2(X_i)$ and $T_2(X_j)$ are node-disjoint, (2) for
each $i$, $T_1|X_i = T_2|X_i$. A \emph{maximum agreement forest} is an agreement forest with a minimum number of components (such that it \emph{maximizes} the agreement), and
this minimum is denoted $d_{MAF}(T_1,T_2)$. In 2001 it was proven by Allen and Steel that $d_{MAF}(T_1, T_2) = d_{TBR}(T_1, T_2) + 1$ \cite{AllenSteel2001}.

In order to define the Maximum Parsimony Distance \cite{fischer2014maximum,moulton2015parsimony,kelk2015reduction} between two unrooted binary phylogenetic trees $T_1, T_2$
both on $X$, we need first to define the concept of \emph{character} on $X$ which is simply a surjection $f:X \rightarrow \mathbf{C}$ where $\mathbf{C}$ is a set of
\emph{states}. Given a tree $T$ on $X$, and a character $f$ also on $X$, an \emph{extension} of $f$ to $T$ is a mapping $f'$ from $V(T)$ to $\mathbf{C}$ such that $f'(\ell) =
f(\ell)$, $\forall \ell \in X$. An edge $e = \{u,v\}$ with $f'(u) \neq f'(v)$ is known as a \emph{mutation} induced by $f'$. The  minimum number of mutations ranging over all
extensions $f'$ of $f$ is called the \emph{parsimony score} of $f$ on $T$ and is denoted by $l_f(T)$.
Given two trees $T_1, T_2$ their \emph{maximum parsimony distance} $d_{MP}(T_1,T_2)$ is equal to $\max_f |l_f(T_1) - l_f(T_2)|$.

Both the TBR and MP distances are \textbf{NP}-hard to compute and they are also \emph{metric} distances i.e., they satisfy the four axioms of metric spaces: (a) non-negativity, (b) 	
identity of indiscernibles (c)	symmetry and (d) triangle inequality.

Given an unrooted binary phylogenetic tree $T$ and a distance $d$ (such as TBR and MP), we define the \textit{unit ball} or the \textit{unit neighborhood} of $T$ under $d$ to be $u_d(T) = \{
T': d(T, T') = 1\}$ i.e., the set of all trees $T'$ that are within distance one from $T$ under the distance $d$.
Such neighborhoods are important because usually
they are building blocks of ``local search" algorithms that try to find trees that optimize some particular criterion.
\twu{Moreover, the diameter $\Delta_n(d)$ is defined as the maximum value $d$ taken over all pairs of phylogenetic trees with $n$ taxa (see~\cite[Section 2.5]{steel2016phylogeny} for a recent review on various results on the unit ball and the diameter of several tree rearrangement metrics.}


\section{Treewidth Distance}

\steven{
The main purpose of this manuscript is to define and study the properties of the \textit{treewidth distance} between two phylogenetic trees.
As mentioned in the introduction, the study of treewidth in the context of phylogenetics was triggered by the pioneering work of Bryant \&
Lagergren \cite{bryant2006compatibility}  who proved that a necessary condition for a set of trees (not necessarily on the same set of taxa) to be compatible, is that their display graph has bounded treewidth. They used this insight to leverage a (theoretical) positive algorithmic result. Here we are interested in the question: in how far does the treewidth
of the display graph \emph{itself} function directly as a measure of phylogenetic incongruence? Hence
the following natural definition:}


\begin{dfn}[Treewidth Distance]
Given two unrooted binary phylogenetic trees $T_1,T_2$, both on the same set of leaf labels $X$, \steven{where $|X| \geq 3$}, their treewidth distance is defined to be $tw(D(T_1,T_2)) - 2$ and is denoted as
$d_{tw}(T_1, T_2)$.
\end{dfn}

It is easy to see that for two unrooted binary phylogenetic trees $T_1,T_2$ we have that $d_{tw}(T_1,T_2) \geq 0$, for $|X| \geq 3$. This is a direct consequence of the fact
that if $|X| \geq 3$ then the display graph contains at least one cycle and hence $tw(D(T_1,T_2)) \geq 2$. If $|X| < 3$ then $T_1,T_2$ are trivially isomorphic (they are either
a single edge or a single vertex) and \steven{it does not make much sense to define a distance between such trees.} So we can discard these boundary cases without any loss of
generality in our study. \steven{(Of course, the treewidth of the display graph is still well-defined in these omitted boundary cases)}. On the other hand we will leverage the
well-known fact that $tw(D(T_1,T_2)) = 2$ for two  unrooted binary phylogenetic trees on $X$, $|X| \geq 3$, if and only if $T_1$ and $T_2$ are compatible (see e.g.
\cite{grigoriev2015}). As mentioned earlier compatibility in this context is the same as label-preserving isomorphism, so it is natural to speak of equality and write $T_1 =
T_2$. Note that it was shown in \cite{kelk2015} that $tw( D(T_1, T_2) ) \leq d_{MAF}(T_1,T_2) + 1 = d_{TBR}(T_1, T_2)+2$, and hence $d_{tw}(T_1,T_2)\leq d_{TBR}(T_1,T_2)$.


\steven{We remark that, because computation of treewidth is fixed parameter tractable \cite{Bodlaender96,downey2013fundamentals}, so too is $d_{tw}$. As we discuss in the final
section of the paper it is not known whether $d_{tw}$ can be computed in polynomial time, but ongoing research efforts by the algorithmic graph theory community to compute
treewidth efficiently in practice (see e.g. \cite{Bodlaender2012}) will naturally strengthen the appeal of $d_{tw}$ as a phylogenetic measure.}



An easy, but important, observation that we will use extensively in the rest of the manuscript is that treewidth (and treewidth distance) are unchanged by edge subdivision and
degree-2 vertex suppression operations - with one trivial exception. We say that a graph is a \emph{unique triangle} graph if  it contains exactly one cycle such that this
cycle has length 3 and at least one of the cycle vertices has degree 2. A unique triangle graph has treewidth 2.

Given a graph $G=(V,E)$, let $e = \{u_1,u_2\} \in E$ be any edge of $G$ and $v$ be any degree-2 vertex of $G$ (if there exists any) with neighbors $v_1,v_2$. We define the
following two operations:

\begin{description}
\item[Subdivision of an edge $e$:] This defines a new graph $G'=(V',E')$ where $V' = V \cup \{w\}, w \notin V$ and $E' = (E \setminus \{e\}) \cup (\{u_1,w \},\{w,u_2\})$.
\item[Suppression of a degree-2 vertex $v$:] This defines a new graph $G''=(V'', E'')$ where $V'' = V \setminus \{v\}$, $E''= E \setminus (\{v_1,v\} \cup \{v,v_2\}) \cup
    \{v_1,v_2 \}$.
\end{description}

\begin{obs}\label{obs:2}
Let $G=(V,E)$ be a graph, which is not a unique triangle graph and let $e = \{u_1,u_2\} \in E$ be any edge of $G$ and $v$ be any degree-2 vertex of $G$ (if any) with neighbors
$v_1,v_2$. Consider the following two graphs:
\begin{enumerate}
\item $G'=(V',E')$ where we obtain $G'$ after a single application of the edge subdivision step on edge $e \in E(G)$, and
\item $G''=(V'', E'')$ where $G''$ is obtained from $G$ after suppressing a degree-2 vertex $v \in V(G)$.
\end{enumerate}
Then we have that:
$$ tw(G) = tw(G') = tw(G'').$$
\end{obs}

\begin{proof}
For the subdivision of an edge case, let $G'$ be the resulting graph after the subdivision of some edge $e$. It is immediate that the treewidth of $G'$ is at least $q = tw(G)$
since $G$ is a minor of $G'$ and treewidth is non-increasing under minor operations. To show that the treewidth cannot increase we argue as follows. If $G$ is a tree, then $G'$
is also a tree and $tw(G) = tw(G') = 1$ and we are done. So, we assume that $G$ is not a tree so $q \geq 2$. Take a bag $B$ of an optimal tree decomposition $\mathbb{T}$ of $G$
with largest bag size at least 3, that contains the endpoints $u_1,u_2$ of $e$. Create a new bag $B' \notin \mathcal{B}(\mathbb{T}): B' = \{u_1,u_2,w\}$ and attach it to $B$.
This operation cannot increase the treewidth of the tree decomposition and it is immediate that the new tree decomposition is a valid one for $G'$.

Now we will handle the degree-2 vertex suppression operation. \steven{This can be simulated by two edge contraction operations, which are minor operations, so the treewidth cannot increase}. In the other direction (i.e. proving that the treewidth cannot decrease), we see that if $G$ is a tree the treewidth is immediately preserved. If $G$ is not a tree, let $G''$ be
the resulting graph after a single degree-2 vertex suppression operation on a vertex $v$ with neighbors, in $G$, $v_1,v_2$ such that in $G''$ $\{v_1,v_2\} \in E''$. Take an
optimal tree decomposition of $G''$, let this be $\mathbb{T}''$. By assumption that $G$ is not a tree and that $G$ is not a unique triangle graph, $G''$ contains at least one cycle. Hence,
 $tw(G'') \geq 2$ i.e., the size of the largest bag is at least 3. In $\mathbb{T}''$, locate a bag $A$ that contains the pair of vertices $v_1,v_2$.
Such a bag must exists by definition. Create a new bag $A' = \{v_1,v,v_2\}$ and attach it to $A$ thus creating a new tree decomposition $\mathbb{T}'''$. It is immediate that
$\mathbb{T}'''$ is a valid tree decomposition for $G$ with width the same as the width of $\mathbb{T}''$, and the claim follows.
\end{proof}

Recall that if two unrooted binary trees $T_1, T_2$ are incompatible, then $tw(D(T_1, T_2)) \geq 3$, so the display graph cannot be a unique triangle graph. Hence:

\begin{obs}
\label{obs:2b}
Let $T_1$ and $T_2$ be two unrooted binary phylogenetic trees on the same set of taxa $X$.
If $T_1$ and $T_2$ are incompatible, then the following operations can be applied
arbitrarily to $D(T_1, T_2)$ without altering its treewidth: suppression of degree-2 vertices, and subdivision of edges.
\end{obs}

In subsequent sections we will often use Observation \ref{obs:2b} to (in particular)
suppress some or all of the taxa in the display graph without altering its treewidth.


\subsection{Metric properties of $d_{tw}$}

Given the definition of the treewidth distance, it is tempting to see if indeed such a distance is a metric distance e.g., it satisfies the four axioms of metric distances. We
already argued that it satisfies the non-negativity condition and trivially it satisfies the identity of indiscernibles because $T_1 = T_2 \Leftrightarrow d_{tw}(T_1,T_2) = 0$
as demonstrated in the previous discussion.  The symmetry condition is also trivially satisfied because $D(T_1,T_2) = D(T_2,T_1)$ i.e., the display graph is identical in both
cases and thus has the same treewidth.

\begin{figure}[ht]
\centering
\includegraphics[scale=0.8]{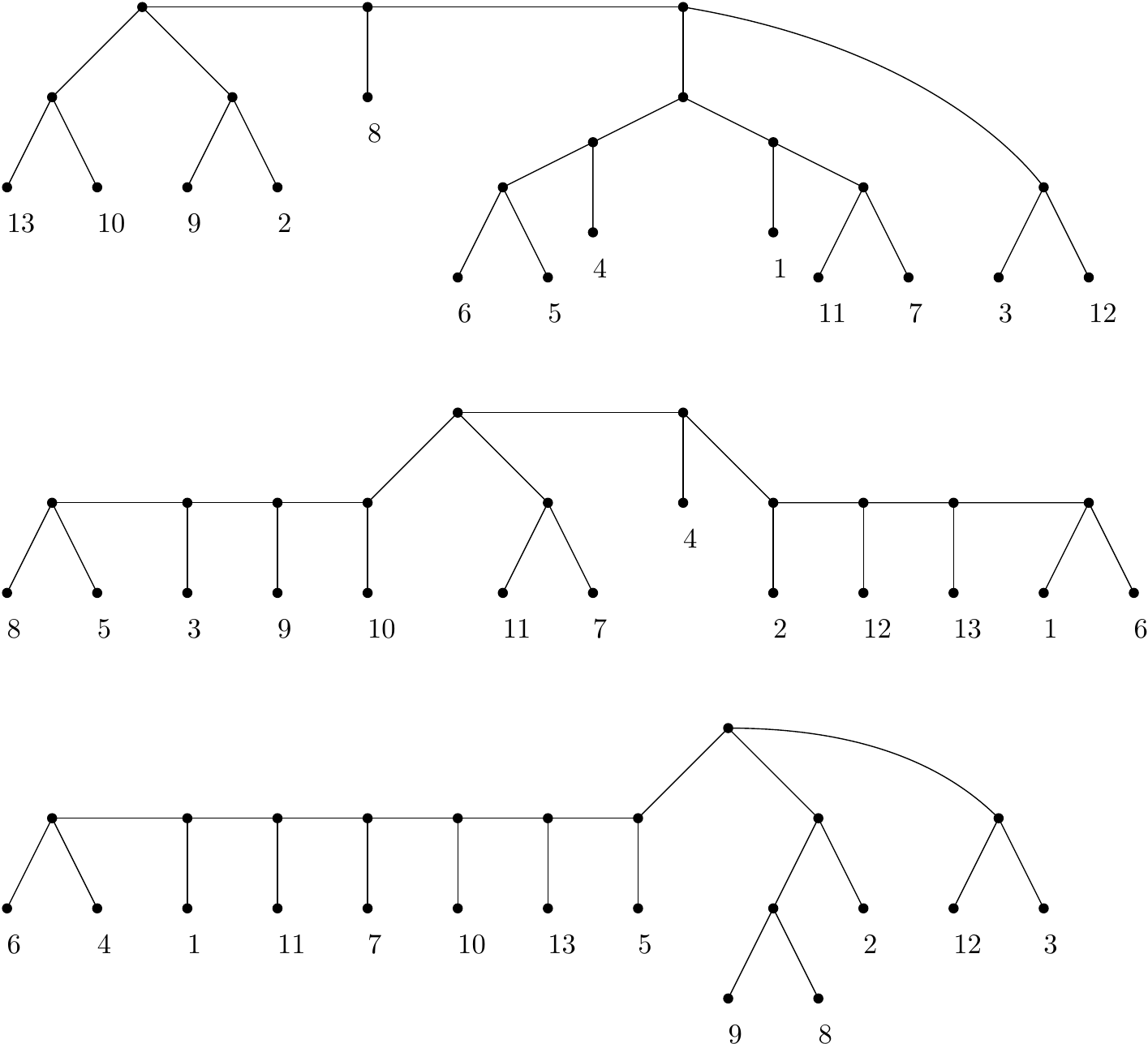}
\caption{An example of three trees \steven{(from top to bottom: $T_1, T_2$ and $T_3$) } on a common set of taxa for which the triangle inequality is violated.}
\label{fig:metric}
\end{figure}

The only case left is to see if $d_{tw}$ satisfies the triangle inequality property: given three unrooted binary phylogenetic trees $T_1,T_2,T_3$ all on $X$ is it the case that
$d_{tw}(T_1, T_3) \leq d_{tw}(T_1,T_2) + d_{tw}(T_2,T_3)$? Unfortunately, this is false as shown in Figure \ref{fig:metric}. By using appropriate software, for example QuickBB
\cite{Gogate:2004:CAA:1036843.1036868}, we can see that  $d_{tw}(T_1,T_2) = 1, d_{tw}(T_2,T_3) = 2$ and $d_{tw}(T_1, T_3) = 4 > d_{tw}(T_1,T_2) + d_{tw}(T_2,T_3)$. \steven{We
remark that, although mathematically disappointing, the absence of the triangle inequality is not a great hindrance in practice. Some other well-known phylogenetic measures,
such as hybridization number, also do not obey the triangle inequality \cite{Semple2007}.}



\section{The treewidth of the display graph under phylogenetic reduction rules}

In this section we investigate the effect of several common phylogenetic reduction rules on the treewidth of the display graph. We will study the following three rules: (i)
common pendant subtree, (ii) common chain and (iii) cluster reduction rule. Such rules constitute the building block of many FPT algorithms for computing phylogenetic
distances. We will see that the three reduction rules behave somewhat differently with respect to the treewidth of the display graph. In particular, we will show how the
subtree reduction operation, where compatible subtrees are collapsed to a single taxon, preserves the treewidth of the display graph. For the second case, the collapsing of a
common chain (a maximal ``caterpillar-like'' region) in both trees down to length 2, could potentially decrease the treewidth of the display graph by \emph{at most} one. On the
other hand we show that if we collapse common chains down to length that is a function of the treewidth of the display graph, then we preserve the treewidth. The open question
here is if this gap can be understood better i.e., if we can collapse the common chains to a constant length and preserve the treewidth. Finally, we investigate the cluster
reduction rule where clusters are formed if in each tree there is an edge (called a \emph{common split}) the deletion of which results that both trees are split into two
subtrees on $X'$ and $X''$. We will see that the treewidth of the display graph is (up to additive terms) equal to the \emph{maximum} of the treewidth of the two clusters. We
note that this is in contrast to other phylogenetic distance measure which usually behave \emph{additively} with respect to the distances of the two clusters.

\steven{It is well-known that compatibility is preserved under
the described reductions. For this reason} we will assume that the two input trees $T_1,T_2$ on $X$ are \emph{not} compatible. This immediately gives us a lower bound on the cardinality of the taxon set, namely $|X| \geq 3$ since any two trees on 2
taxa are by definition compatible (both trees are single edges). Moreover the treewidth of their display graph is at least 3.

We start with the common pendant subtree rule.

\subsection{Subtree Reduction Rule}
Let $T_1, T_2$ be two unrooted binary phylogenetic trees on the same set of taxa $X$. A  subtree $T$ is called a \textit{pendant} subtree of $T_i$, $i \in \{1,2\}$ if there exists an edge $e$
the deletion of which detaches $T$ from $T_i$. A subtree $T$, which induces a subset of taxa $X' \subset X$, is called \textit{common pendant subtree} of $T_1$ and $T_2$ if
$T_1 | X' = T_2 | X'$ and if the additional following condition holds:

\begin{itemize}
\item[$\triangleright$] Let $e_i$ be the edge of tree $T_i, i \in \{1,2\}$ the deletion of which detaches $T$ from $T_i$ and let $v_i \in e_i, i \in \{1,2\}$ be the endpoint
    of $e_i$ ``closest" to the taxon set $X'$. Let's say that we root each $T_i|X'$ at $v_i$, thus inducing a \emph{rooted} binary phylogenetic tree $(T_i|X')^{\rho}$ on
    $X'$. We require that $(T_1|X')^{\rho} = (T_2|X')^{\rho}$.
\end{itemize}

The previous condition formalizes the idea that the point of contact of the pendant subtree with the rest of the tree should explicitly be taken into account when determining whether a pendant subtree is common.
(This is consistent with the definition of common pendant subtree elsewhere in the literature). 

\begin{figure}[ht]
\centering
\includegraphics[scale=0.8]{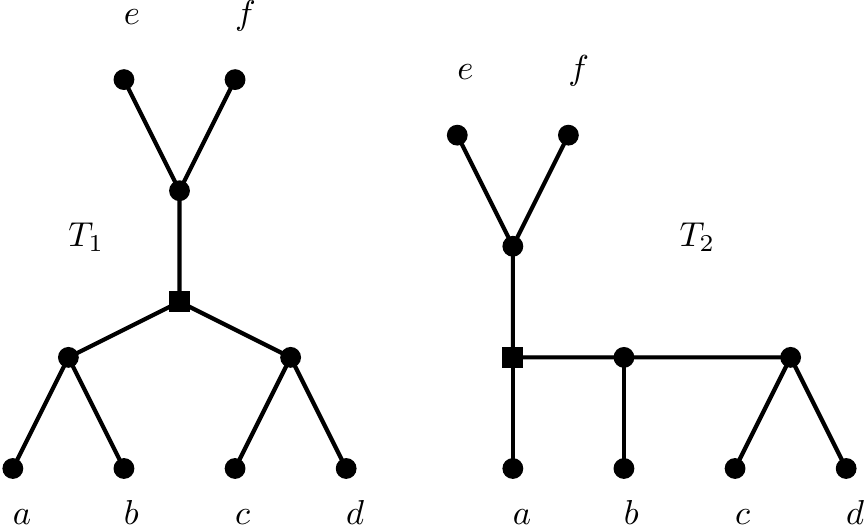}
\caption{An illustration of the concept of \emph{common pendant subtree}. If $X' = \{a,b,c,d\}$ we see that $T_1|X' = T_2|X'$ (because of the suppression of the parent of $a$ in $T_2$) but this is not true if we take into account the
``root" location (in bold squares). \steven{Here both $\{e,f\}$ and $\{c,d\}$ induce maximal common pendant subtrees.}}
\label{fig:cherry}
\end{figure}

In the following we will show that the treewidth of the display graph $D(T_1,T_2)$ of the two phylogenetic \steven{trees} $T_1,T_2$ is preserved under the common pendant subtree
reduction rule:

\begin{description}
\item[Common Pendant Subtree (CPS) reduction:] Find a maximal common pendant subtree in $T_1,T_2$. Let $T$ be such a common subtree with at least two taxa and let $X_T$ be
    its set of taxa. Clip $T$ from $T_1$ and $T_2$. Attach a single label $x \notin X$ in place of $T$ on each $T_i$. Set $X := (X \setminus X_T) \cup \{x\}$ and let $T_1',
    T_2'$ be the two resulting trees and $D(T_1',T_2') = D'$ be their resulting display graph.
\end{description}

\begin{thm}
\label{thm:subtree:reduction} Suppose that $T_1$ and $T_2$ are a pair of incompatible unrooted binary phylogenetic trees on $X$ and the pair $(T'_1,T'_2)$ is obtained from $(T_1,T_2)$ by one
application of the Common Pendant Subtree reduction.  Then $d_{tw}(T_1,T_2)=d_{tw}(T'_1,T'_2)$.
\end{thm}

\begin{figure}[ht]
\centering
\includegraphics[scale=1]{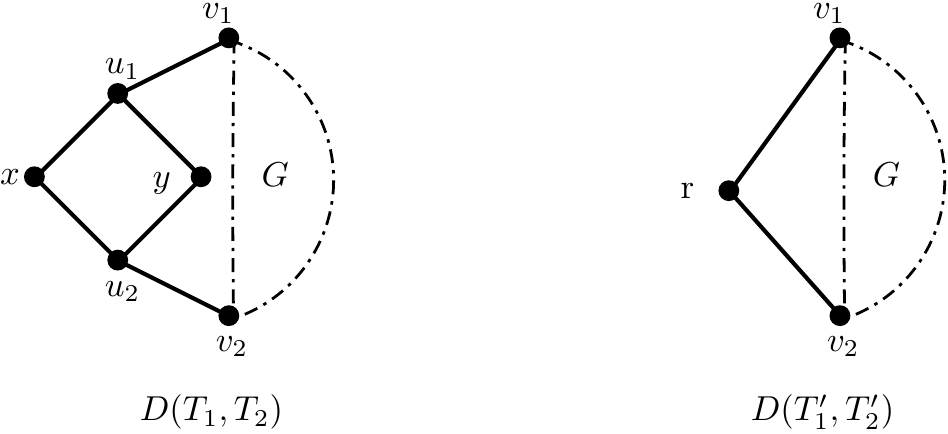}
\caption{Reduction of a common cherry $\{x,y\}$ as described in the proof of Theorem~\ref{thm:subtree:reduction}.}
\label{fig:cherry}
\end{figure}

\begin{proof}
A \emph{cherry} is simply a size-2 subset of taxa $\{x,y\}$ that have a common parent, and
a cherry $\{x,y\}$ is \emph{common} if it is in both trees.
Let us first consider the case that the pair $(T'_1,T'_2)$ is obtained from a subtree reduction on a common cherry $\{x,y\}$ whose parent is $u_i$ in $T_i$ and the
parent of $u_i$ is $v_i$, $i = 1,2$. Then the display graph $D'=D(T'_1,T'_2)$ is obtained from $D=D(T_1,T_2)$ by replacing the vertex subset $\{u_1,x,y,u_2\}$ with a single
vertex $r$ which is connected to $v_1$ and $v_2$ and these are the only neighbors of $r$ (see Figure \ref{fig:cherry}). Note that \steven{$v_1 \neq v_2$ and $\{v_1,v_2\}$ is not an edge in $D$ because $T_1$ and $T_2$ are incompatible.}
$D'$ can be obtained from $D$ by applying Observation \ref{obs:2b}: suppress $x$, suppress $y$ (and delete the created multi-edge) and then suppress $u_2$. Hence $tw(D') = tw(D)$. (The surviving vertex $u_1$ assumes the role of $r$, since labels are irrelevant to treewidth.)
For the more general case: it is easy to see that applying the CPS reduction rule to a subtree that is not a cherry, can be achieved by iteratively applying the CPS reduction to common cherries. This is correct because collapsing a common cherry cannot make two incompatible trees compatible. The result follows.
\end{proof}

\subsection{Chain Reduction Rule}

Let $T$ be an unrooted binary tree on $X$. For each taxon $x_i \in X$, let $p_i$ be its unique parent in $T$. Let $C = (x_1, x_2, \dots, x_t)$ be an ordered sequence of taxa
and let $P = (p_1, p_2, \dots, p_t)$ be the corresponding ordered sequence of their parents,
If $P$ is a \textit{path} in $T$ and the $p_i$ are all mutually distinct then $C$
is called a \textit{chain} of length $t$. A chain $C$ is a \textit{common chain} of two binary phylogenetic trees $T_1, T_2$ on a common set of taxa, if $C$ is a chain in each
one of them. See Figure \ref{fig:two_ns} for an example. \steven{Note that our insistence that the $p_i$ are mutually distinct differs from the definition of chain encountered elsewhere in the literature, in which
$p_1 = p_2$ and $p_{t-1} = p_t$ is permitted. However, our more restrictive definition of
chain is only a very mild restriction, since a chain of length $t$ under the traditional
definition yields a chain of length at least $(t-4)$ under our definition. Our definition ensures
that in both trees neither end of the chain is a cherry, which avoids a number of annoying
(and uninteresting) technicalities.} Let $v_i$ denote the parent of $x_i$ in $T_1$ and
$u_i$ its parent in $T_2$.

 We now define the common chain reduction rule.

\begin{description}
\item[Common $d$-Chain Reduction Rule ($d$-cc):] Let $T_1, T_2$ be two \steven{incompatible} unrooted binary phylogenetic trees on a common set of taxa $X$. Let $C$ be a common chain of $T_1, T_2$ of
    length $t \geq 3$. On each $T_i, i \in \{1,2\}$ clip the chain down to length $d \in \{2, \ldots, t-1\}$ as follows: Keep the first $\lceil d/2 \rceil$ and the last
    $\lfloor d/2 \rfloor$ taxa
and delete all the
    intermediate ones (i.e., delete all the taxa with indexes in $\{\lceil d/2 \rceil +1, \ldots,  t - \lfloor d/2 \rfloor
\}$), suppress any resulting vertex of degree $2$ and
    delete any resulting unlabelled leaves of degree 1. Let $C'$ be the new clipped common chain on both trees.
\end{description}

Observe that $C'$ has $\lceil d/2 \rceil + t -(t - \lfloor d/2 \rfloor) = \lceil d/2 \rceil + \lfloor d/2 \rfloor = d$ taxa and that in each $T_1, T_2$ the parents of the taxa
$x_{\lceil d/2 \rceil}, x_{t - \lfloor d/2 \rfloor + 1}$ are connected by an edge.  Let $D(T_1,T_2) = D$ be the display graph of $T_1,T_2$ and $D'(T_1',T_2') = D'$ be the display
graph of $T_1',T_2'$ after the application of one chain reduction rule. Equivalently, $D'$ can be obtained \steven{directly from $D$ by deleting the $(t-d)$ pruned taxa and suppressing
unlabelled degree-2 vertices}.
\steven{In fact, due to the fact that $T_1, T_2$ are incompatible (and thus so are $T'_1, T'_2$)
we can (by Observation \ref{obs:2b}) safely suppress (in $D$) all the degree-2 nodes labelled by
taxa in $C$, and (in $D'$) all the degree-2 nodes labelled by taxa in $C'$, without altering
the treewidth of $D$ or $D'$.  Without loss of generality we assume
that this suppression has taken place.}

Observe that the part of $D$ that corresponds to the common chain $C$ \steven{now} resembles a $2 \times t$ grid and
in $D'$ is a $2 \times d$ grid. For a common chain $C$ of length $t$, let $g(C)$ be the corresponding $2\times t$ grid in $D$ and similarly define $g(C')$ in $D'$ for the
clipped common chain of length $d$.

Now, assume that we have an optimal tree decomposition $\mathbb{T}$ of $D$ of width $k$, i.e., the maximum bag size in $\mathbb{T}$ is $k+1$. First of all, by a
standard minor argument, it is immediate that application of the cc-reduction rule cannot increase the treewidth: the resulting display graph $D'$ is a minor of $D$.

Our strategy will be as follows:
Given an optimal tree decomposition $\mathbb{T}'$ for $D'$, we will modify it to construct
a tree decomposition for $D$ that in the worst case has width at most $tw(D')+1$, thus
proving $tw(D') \geq tw(D) - 1 = k - 1$. (In some cases we will be able to prove the stronger result that $tw(D) = tw(D')$).






We distinguish two cases.


\smallskip
\noindent \textit{\textbf{Case 1: The common chain $g(C)$ is a separator in D.}} In other words, deleting $g(C)$ from $D$ will result in two connected components. In this case
we will show that clipping the common chain $C$ down to length $2$ by applying a $2$-cc step preserves the treewidth of $D$. We note that an application of a $2-cc$ step causes
$g(C')$ to resemble a $C_4$ in $D'$, where as usual, $C_4$ is a cycle of length 4.

\begin{lem}\label{no_sep}
Let $T_1',T_2'$ be two incompatible unrooted binary phylogenetic trees  that are obtained after a single application of the operation $2$-cc$(C)$ on $T_1$ and $T_2$ where $g(C)$  is a separator in $D(T_1,T_2)$. Then $d_{tw}(T_1,T_2)=d_{tw}(T'_1,T'_2)$.
\end{lem}

\begin{proof}
Let $D$ be the display graph of $T_1, T_2$ and $D'$ the display graph after we clipped the common chain $C$ down to length 2 and let $g(C')$ be the $2 \times 2$ grid induced by
the common chain in $D'$. Remember that $g(C')$ has 4 vertices $\{ v_1, u_1, v_t, u_t \}$ such that $\{v_1,v_t\} \subset V(T_1)$ and $\{u_1,u_t\} \subset V(T_2)$. Let $\mathbb{T}'$
be an optimal tree decomposition for $D'$.


Consider the grid $g(C')$ in $D'$ corresponding to the clipped chain $C'$ of length $d =2$.  We will  expand $g(C')$
inductively by first inserting the parents $v_2, u_2$ of the clipped taxon $x_2$ \steven{(and an edge between them)}: These two vertices will be inserted in the $C_4$ induced by $\{ v_1, v_{t}, u_1, u_{t} \}$. \steven{After}
the $j$-th step, $j \leq t-d$, of this process, we \steven{will} have retrieved the parents of taxa $x_2, \dots x_{j+1}$. Step $(j+1)$ continues by expanding the current $g(C'')$ of length \steven{$j+2$} by inserting
the parents \steven{$v_{j+2}, u_{j+2}$} in the $C_4$ induced by $v_{j+1}, v_{t}, u_{j+1}, u_{t}$. We will show how, at each step, we can update the tree decomposition $\mathbb{T}'$, without increasing its width, so that the new one
will be a valid tree decomposition for the updated display graph.

We will start by proving the base case. For this, we will find helpful the following claim about the structure of $\mathbb{T}'$.

\begin{claim}
There exists an optimal tree decomposition $\mathbb{T}'$ of $D'$ such that $\mathbb{T}'$ contains two adjacent degree-2 bags $A_1$ and $A_2$ where \steven{$A_1=\{v_1,u_1,v_t\}$,
$A_2=\{v_t,u_1,u_t\}$}.
\end{claim}

\begin{proof}
Observe that since $g(C)$ is a separator in $D$, then so is $g(C')$ in $D'$. In $D'$ we delete the edges $\{v_1, v_t\}$ and $\{u_1, u_t\}$ and we obtain, wlog, two connected
components $D_1'$ and $D_2'$ such that $\{v_1, u_1\} \subset V(D_1')$ and $\{v_t, u_t\} \subset V(D_2')$. Consider optimal tree decompositions $t_1$, $t_2$ of $D_1', D_2'$
respectively. \steven{Note that $tw(D_1') \leq tw(D'), tw(D_2') \leq tw(D')$ and $tw(D') \geq 3$}. Since $\{v_1, u_1\} \in E(D_1')$, there must be a bag $B_1 \in V(t_1)$ that contains $\{v_1, u_1\}$. Similarly, there must be a bag $B_2 \in V(t_2)$ that contains
$\{v_t, u_t\}$. Attach to $B_1$ a new bag $A_1 = \{ v_1, u_1, v_t \}$ and attach to $B_2$ bag $A_2 = \{v_t, u_1, u_t\}$ and join $A_1, A_2$ by an edge to create a new tree
decomposition $\mathbb{T}'$ for $D'$: indeed, it is immediate to see that $\mathbb{T}'$ satisfies all the treewidth conditions. Moreover, the width of this tree decomposition is
$\max (tw(D_1'), tw(D_2'), 2 )$. Noting that
$3 \leq tw(D') \leq \max (tw(D_1'), tw(D_2'), 2 ) \leq max (tw(D_1'), tw(D_2')) \leq tw(D')$
it follows that it is an optimal tree decomposition of $D'$.
\end{proof}

\begin{figure}[ht]
\centering
\includegraphics[scale=0.8]{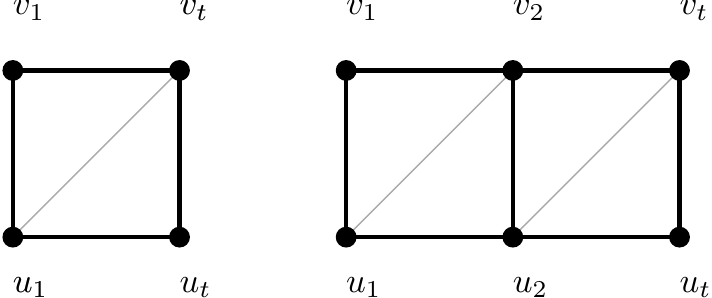}
\caption{An example of the inductive construction of Lemma \ref{no_sep}. We construct a new tree decomposition which facilitates the extra links added by increasing the length of $g(C')$ by one, corresponding to adding the parents of the current missing taxon (in this case $x_2$). \steven{The grey edges are not in the display graph $D'$ but they indicate the maximal cliques induced by
the size-3 bags that we add in Lemma \ref{no_sep}.}}
\label{fig:no_sep}
\end{figure}

Given $\mathbb{T}'$ as described in the previous claim, delete bags $A_1, A_2$ and consider the following set of bags: $J_1 = \{ v_1,v_2, u_1 \}$,
$J_2 = \{ v_2,u_1, u_2 \}$, $J_3 = \{ v_2,v_t, u_2 \}$ and $J_4 = \{ v_t, u_2, u_t \}$. Attach $J_1$ to $B_1$ (the bag that was adjacent to $A_1$) and $J_4$ to $B_2$ (the bag
that was adjacent to $A_2$) and create a path of bags from $J_1$ to $J_4$. It is easy to argue that this is a valid tree decomposition $D''$, defined as the display graph after the parents of $x_2$ have been added; see Figure \ref{fig:no_sep}. First of all, for conditions (tw1) and (tw2) this is immediate by construction. Indeed, $v_2$ belongs to $J_1,J_2,J_3$ and $u_2$ belongs
to $J_2,J_3,J_4$. For (tw2) observe that the edges $\{v_1, v_t\}, \{u_1, u_t\}$ are not present in $g(C'')$ so we do not need to consider them. For the new edges we have that
$\{ v_1,v_2 \} \in J_1$, $\{u_1,u_2\} \in J_2$, $\{v_2, u_2\} \in J_3$, $\{v_2,v_t\} \in J_3$ and $\{u_2,u_t\} \in J_4$. Also, by leveraging the explicit construction of $\mathbb{T}'$ (in particular: $v_t, u_t \not \in B_1$ and $u_1, v_1 \not \in B_2$) we can easily verify that (tw3) is true for $\mathbb{T}''$.
Finally, the width of this new tree decomposition is no greater than the width of $\mathbb{T}'$ because we only add bags of size 3 and, by construction,  $\mathbb{T}'$
already contained at least one bag of size 4.

This proves that, for the base case, the treewidth of the new display graph remains unchanged. For the $j$-th step, we apply the arguments above where as $A_1$ and $A_2$ we use the bags
$\{ v_j, u_j, v_t  \}$ and $\{ v_t, u_j, u_t \}$ which by induction exist and are adjacent.
Delete them and replace them with the following chain of bags, as before: $J_1 = \{ v_j,v_{j+1}, u_j \}$, $J_2 = \{ v_{j+1},u_j, u_{j+1} \}$,
$J_3 = \{ v_{j+1} , v_t, u_{j+1} \}$ and $J_4 = \{ v_t, u_{j+1}, u_t \}$. We continue until we add the last missing piece of $g(C)$.
\end{proof}

\noindent \textit{\textbf{Case 2: The common chain $C$ is not a separator in D.}} 
We say that the $2 \times t$ grid $g(C)$ in $D$ that corresponds to the common chain $C$ is not a separator if the deletion of $g(C)$ from $D$ leaves the display graph $D$
connected. See Figure \ref{fig:two_ns} as an example of such a case and Figure \ref{fig:dis_graph} for an example of their display graph. It is easy to observe that if $g(C)$
is not a separator in $D$ then neither is $g(C')$ in $D'$. We will show that in this case the treewidth of $D$ after clipping $g(C)$ down cannot decrease by more than a unit
term.

\begin{lem}\label{chain_sep}
Let $T_1',T_2'$ be the two incompatible unrooted binary phylogenetic trees  that are obtained after a single application of the $d$-cc reduction rule with $d=2$ on $T_1$ and $T_2$ on a common chain $C$
such that $g(C)$ is \emph{not} a separator in $D(T_1,T_2)$. Then we have $d_{tw}(T'_1,T'_2) \leq d_{tw}(T_1,T_2) \leq d_{tw}(T'_1,T'_2)+1.$
\end{lem}

\begin{proof}
As in the separator case, we will alter the tree decomposition $\mathbb{T}'$ for $D'$ to obtain a new tree decomposition $\mathbb{T}''$ that will be valid for $D''$ (the display graph with the expanded $2 \times 3$ grid $g(C'')$)  and which has width at most $tw(D')+1$. Then, we will argue how we can increase the length of this $2 \times 3$ grid $g(C'')$ to any
arbitrary length without further increasing the width. So, the $+1$ term might be incurred only when we transfer from the $2 \times 2$ to the $2 \times 3$ grid but when we retrieve the
rest of $C$ we do not have to pay again in terms of increasing the width. The reason for this is that in the transition from length 2 to 3 we guarantee that the tree
decomposition for the updated situation has a certain invariant property that we can exploit in order to further increase the length of the grid ``for free''.  The initial tree
decomposition might however not possess this property and we have to pay potentially a unit increase in the width of the decomposition to establish it.

Consider the grid $g(C')$ in $D'$ corresponding to the clipped chain $C'$ of length $d =2$. It contains 4 vertices: $\{ v_1,  v_{t} \} \in V(T_1)$ and $\{ u_1, u_{t} \} \in
V(T_2)$. As in the separator case we will expand this $g(C')$ inductively by first inserting the parents $v_2, u_2$ of the clipped taxon $x_2$ and after the $j$-th step, $j \leq t-d$ of this process we will have already retrieved the parents of taxa $x_2, \dots x_{j+1}$. The $(j+1)$th step proceeds by expanding the current $g(C'')$ of length $j+2$ by inserting the parents $v_{j+2}, u_{j+2}$ in the $C_4$ induced by $v_{j+1}, v_{t}, u_{j+1}, u_{t}$.

For the base case,  we will distinguish \steven{three cases. In all cases we assume without loss of generality that
$\mathbb{T}'$ is an optimal \emph{small} tree decomposition of $D'$. A small tree decomposition is a tree decomposition
where no bag in the tree decomposition is a subset of another (which thus also excludes the possibility of having two
copies of the same bag). It is well-known that there exist optimal tree decompositions that are also small.}

\begin{description}
\item[$|V(D')| > 4$ and $\exists$ bag $B \in V(\mathbb{T}')$ such that $B$ contains $\{v_1, v_{t}, u_1, u_{t}\}$.] As a first step, we claim that $|B| \geq 5$.
Indeed, assume for the sake of contradiction that $B$ contains only these four vertices and take any bag $A \in V(\mathbb{T}')$ that is adjacent to $B$ in the tree decomposition $\mathbb{T}'$. \steven{(Such a bag must exist because $|V(D')| > 4$.)}  Consider their
    intersection $A \cap B$. By the smallness assumption on $\mathbb{T}'$ we have that $|A \cap B| \leq 3$.  By standard properties of tree
    decompositions (see e.g., \cite{Cygan:2015:PA:2815661}) we know that $A \cap B$ is a \textit{separator} in $D'$ of the following two sets of vertices: $F_A = \cup_{v \in
    V(T_A)} B_v , F_B = \cup_{v \in V(T_B)} B_v$ where $T_A$ is the connected component of $\mathbb{T}'$ that contains bag $A$ and $T_B$ is the connected component of
    $\mathbb{T}'$ that contains $B$ if we delete the edge $\{A,B\}$ from $E(\mathbb{T}')$. But observe that $A \cap B$ cannot be a separator for separation $F_A, F_B$ because
    $A \cap B \subset g(C')$ and $g(C')$ is not a separator of $D'$. A contradiction.

    Now we proceed as follows:  Create a new bag $H_1 = \{v_1, v_{t}, u_1, u_{t}, v_2\}$
and attach it to $B$ with an edge. Create a second bag $H_2 = H_1 \cup \{u_2\} \setminus \{v_1\}$ and attach it to $H_1$.


    We claim this is a valid tree decomposition for $D''$ (which is $D'$ where $g(C')$ has increased its length by 1). Indeed, property (tw1) is immediate by construction, as
    is (tw3). For (tw2) observe that bag $H_1$ takes care of the new edges $\{ v_1, v_2 \}, \{ v_2, v_{t}\}$ of $g(C'')$ and the bag $H_2$ of the new edges $\{v_2, u_2\}, \{
    u_1, u_2 \}, \{ u_2, u_{t}\}$. Note that, because $|B| \geq 5$, the new bags $H_1$ and
$H_2$ do not increase the width of the decomposition.


\item[$|V(D')| = 4$ and $\exists$ bag $B \in V(\mathbb{T}')$ such that $B$ contains $\{v_1, v_{t}, u_1, u_{t}\}$.] This situation can only occur if $D'$ is the complete graph on 4 vertices $K_4$ (since we know $tw(D') \geq 3$). This
exceptional case can be dealt with similarly to the previous case, except that the addition of
bags $H_1$ and $H_2$  increase the width of the decomposition by exactly
one.  That is, we obtain a decomposition of $D''$ of width $tw(D') + 1$.

\item[$\not \exists B \in V(\mathbb{T}')$ that contains all of $\{v_1, v_{t}, u_1, u_{t}\}$.] Note that
every chordal completion of $D'$ must introduce the chord $\{v_1, u_t\}$ and/or the chord $\{v_t, u_1\}$. It is well-known that each maximal clique in a chordal completion induces
a bag in a corresponding tree decomposition, and each bag in a tree decomposition induces a maximal clique in a corresponding chordal completion. Assume without loss of generality that the chord  $\{v_t, u_1\}$ is present\footnote{If $\{v_1, u_t\}$ is present and not $\{v_t, u_1\}$ then by topological symmetry of the chain the argument still goes through: conceptually we are then simply reconstructing the chain in the ``opposite'' direction.} Then $\{v_1, u_t\}$ is not present (because otherwise the corresponding bag would contain all of  $\{v_1, v_{t}, u_1, u_{t}\}$, violating the case assumption.)
Hence there exist two
    bags $A \neq B$ of $\mathbb{T}'$ that contain the sets of vertices $\{v_1, u_1, v_{t}\}$ and $\{u_1, u_{t}, v_{t}\}$ respectively (and possibly other vertices).   Add the
    element $v_1$ to $B$ and, in order to guarantee the running intersection property for $v_1$, add it also to each of the bags in the unique path from $A$ to $B$ in the tree
    decomposition $T'$ (all these bags contain $\{ u_1, v_{t} \}$ by the running intersection property).   This might increase the width of the decomposition by at most one. We introduce
$H_1$ next to $B \cup \{v_1\}$ and $H_2$ next to $H_1$.
\begin{itemize}
\item  If adding $v_1$ \emph{does} increase the width, it is because $v_1$ is added to a bag that already has maximum size. All
maximum-size bags in $\mathbb{T}'$ contain at least 4 vertices (because $tw(D') \geq 3$) so after adding $v_1$
the maximum-size bags in the decomposition contain at least 5 vertices. Specifically,  adding $H_1$ and $H_2$ cannot further increase the width of the decomposition and we obtain
a decomposition of width at most $tw(D') + 1$.

\item If adding $v_1$ does \emph{not} increase the width, then the maximum bag size in our new $v_1$-augmented
decomposition is at least 4 (because $|B \cup \{v_1\}| \geq 4$). Hence, adding $H_1$ and $H_2$ cannot increase the
width of the decomposition by more than 1. So we again have a decomposition of width at
most $tw(D') + 1$.
\end{itemize}

\end{description}


In all the above three cases we end up with a (not necessarily optimal) tree decomposition
in which $H_1$ and $H_2$ are two adjacent size 5 bags (of degree 2 and 1 respectively).
This process can now be iterated without further raising the width of the decomposition because
all added bags will have size at most 5. For example, to add the parents of $x_3$: add a new bag $\{v_2, u_2, v_t, u_t\}$
next to $H_2$ (``forget'' $u_1$ from bag $H_2$) and then add two new bags
$\{v_2, v_3, v_t, u_2, u_t\}$ (``introduce'' $v_3$) and $\{v_3, v_t, u_2, u_3, u_t\}$
(``forget'' $v_2$ and ``introduce'' $u_3$).


In conclusion, from a clipped chain $C'$ and its corresponding grid $g(C')$ in $D'$ we can retrieve the whole original chain by increasing the treewidth of the resulting
display graph by at most 1. Equivalently, clipping a common chain down to length 2 where in the display graph $D(T_1,T_2)$ the common chain is not a separator, cannot decrease
the treewidth of the resulting display graph by more than 1.
\end{proof}

 \steven{Figure \ref{fig:dis_graph} shows that shortening a chain to length 2 might indeed
reduce the treewidth of the display graph by 1. A natural question therefore arises: is there a constant $d > 2$ such that, if we clip a chain down to length $d$, the treewidth of the display graph is guaranteed to not decrease?} This seems like a highly non-trivial question with
deep connections to forbidden minors. But, at least in the case where the common chain is very large with respect to a function of the treewidth of the display graph $D$, we can
show that shortening chains to a length dependent on the treewidth of $D$ \emph{does} preserve the treewidth.

\begin{thm}
Let $T_1, T_2$ be two incompatible unrooted binary trees and $D(T_1,T_2)$ their display graph such that $tw(D(T_1,T_2)) = k$. Then, there is a function $f(k)$ such that if there exists a
common chain $C$ of length $t > f(k)$ then we can clip $C$ down to length $f(k)$ such
that $tw(D') = tw(D)$ (where as usual $D'$ is the display graph of the trees with the shortened
chains).
\end{thm}

\begin{proof}
Give that $tw(D(T_1,T_2)) = k \geq 3$ we can as usual without loss of generality
suppress all taxa in the display graph. Now, $D(T_1,T_2)$ must have as a minor one of the forbidden minors for treewidth $k-1$. \steven{Forbidden minors for treewidth $k-1$ (where $k-1 \geq 2$) are all connected simple graphs with minimum degree 3.} By the work of Lagergren
\cite{DBLP:journals/jct/Lagergren98}  we know that the number of edges (and vertices) in forbidden minors for treewidth $k$ is bounded by a function $f'$ of $k$ which is doubly exponential in
$O(k^5)$.  Let $d' = f'(k-1)$. Now, fix the image of a forbidden
minor for treewidth $k-1$ inside $D$. Each vertex $v$ of the minor has degree at most $d'$, and
(crudely) a degree $d'$ vertex $v$ can be split into at most $\leq d'$ degree-3 vertices on the image inside $D$ (these are the vertices which via edge contractions will merge to form $v$). Hence a common chain longer than $(d')^2$ must necessarily contain ever more vertices which are not on the image at all, or which are degree-2 vertices on the image. For a sufficiently large function $f$
the point is reached that, if the chain is longer than $f((d')^2)$, reducing the length
of the chain by 1 cannot destroy the forbidden minor: either the image survives or a
slight modification of it (with fewer degree-2 vertices) can be embedded in the graph. Hence,
shortening the chain to length $f((d')^2)$ cannot reduce the treewidth below $k$.
%
%
\end{proof}

\begin{figure}[ht]
\centering
\includegraphics[scale=1]{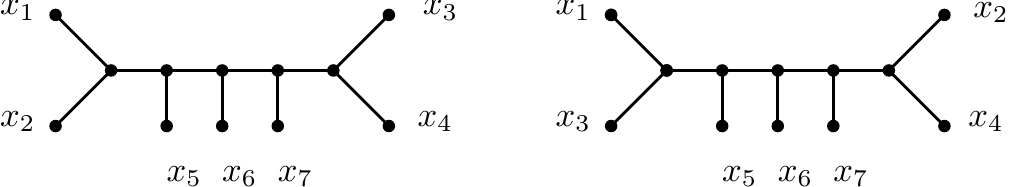}
\caption{An example of two trees with a common chain indexed by taxa $x_5,x_6,x_7$.}
\label{fig:two_ns}
\end{figure}

\begin{figure}[ht]
\centering
\includegraphics[scale=0.8]{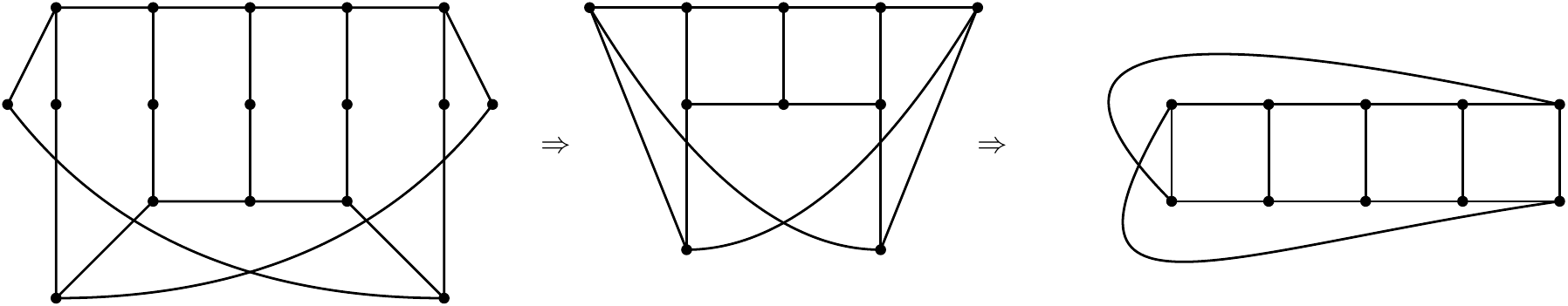}
\caption{The display graph (after the suppression of all vertices of degree two) of the two trees $T_1,T_2$ from Figure \ref{fig:two_ns}.
Observe that the final graph is in fact one of the forbidden minors for treewidth 3, the Mo\"{e}bius graph, and $tw(D) = 4$.
Observe also that if we clip the common chain \steven{down to length 2}, then the treewidth of $D$ decreases to 3.}
\label{fig:dis_graph}
\end{figure}

\subsection{Cluster Reduction Rule}
In this subsection we will study how the treewidth of the display graph relates to the treewidth of its \textit{clusters} which are related to common splits:

\begin{dfn}
Let $T_1$ and $T_2$ be two unrooted binary phylogenetic trees on the same set of taxa $X$. We say that $T_1$ and $T_2$ have a \emph{common split} $X^{*} | X^{**}$ if $X^{*}$
and $X^{**}$ together form a bipartition of $X$ and, for $i \in \{1,2\}$, $T_i$ has some edge $e_i$ such that deleting $e_i$ separates $X^{*}$ from $X^{**}$ in that tree.
\end{dfn}

In the following proofs we will refer extensively to Figure \ref{fig:clusterreduction}.


\begin{figure}[h]
\centering
\includegraphics[scale=0.75]{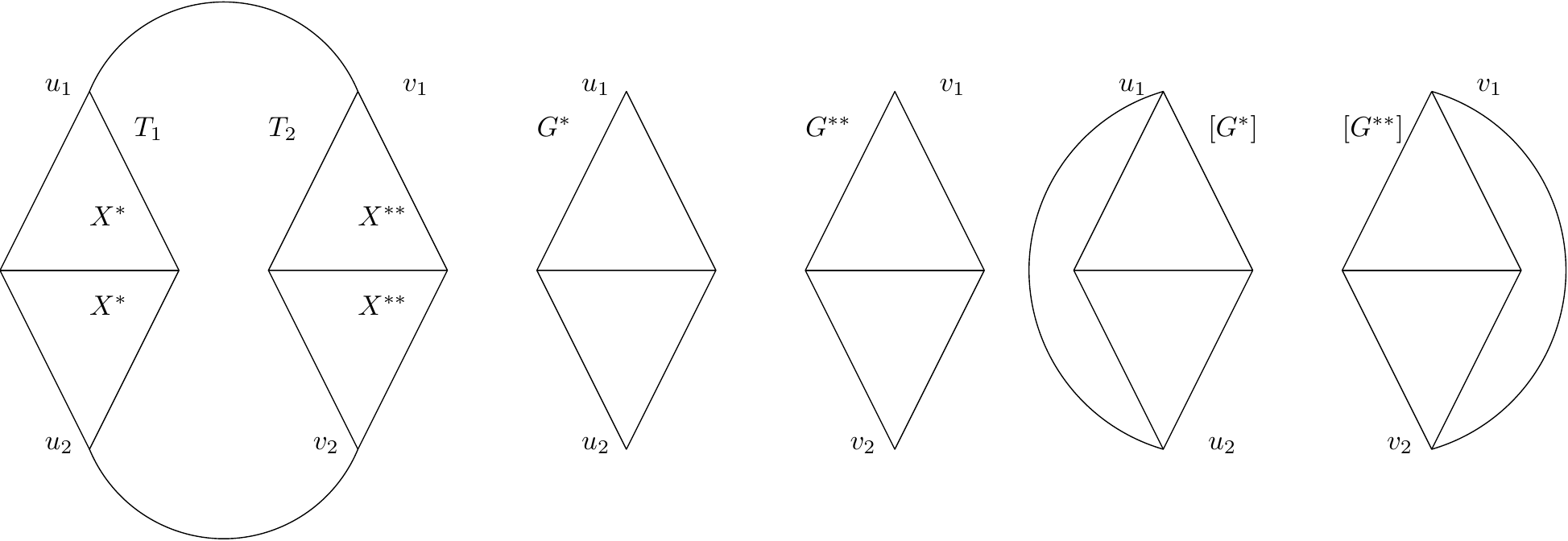}
\caption{Left: the display graph when $T_1$ and $T_2$ have a common split $X^{*}|X^{**}$. Centre: the graphs $G^{*}$ and $G^{**}$ obtained by deleting the two edges inducing the common split. Right: the graphs $[G^{*}]$ and $[G^{**}]$ obtained from $G^{*}$ and $G^{**}$ by joining the ``roots'' together.}
\label{fig:clusterreduction}
\end{figure}

\begin{lem}
\label{lem:cr_bound} Let $T_1$ and $T_2$ be two incompatible unrooted binary phylogenetic trees on the same set of taxa $X$ and let $X^{*}|X^{**}$ be a common split of $T_1$ and $T_2$.
Let $p =tw( D( T_1 | X^{*}, T_2 | X^{*} ) )$ and $q = tw( D( T_1 | X^{**}, T_2 | X^{**} ) )$. Then $$\max(p,q) \leq tw(D(T_1, T_2)) \leq \max(p,q) + 1 $$.
\end{lem}

\begin{proof}
First we observe that the lower bound $\max(p,q) \leq tw(D(T_1, T_2))$ is immediate, since both $D( T_1 | X^{*}, T_2 | X^{*} )$ and $D( T_1 | X^{**}, T_2 | X^{**} )$ are minors
of $D( T_1, T_2)$.

For the upper bound, we will first deal with the case when $|X^*|,|X^{**}| \geq 3$. Let $e_1 = \{u_1,v_1\}$ be the edge that
induces the $X^{*}|X^{**}$ split in $T_1$, and let $e_2 = \{u_2, v_2\}$ be the edge which induces the split in $T_2$. If we delete both the edges $\{u_1, v_1\}$ and $\{u_2,
v_2\}$ from $D(T_1, T_2)$ then we obtain a graph with two connected components. Each one of these two components has two degree-2 vertices, the endpoints of the two deleted
edges.  One of these components is a ``rooted'' version of $D( T_1 | X^{*}, T_2 | X^{*} )$, which we call $G^{*}$, and the other is a ``rooted'' version of $D( T_1 | X^{**},
T_2 | X^{**} )$, which we call $G^{**}$ where, in contrast with $D( T_1 | X^{*}, T_2 | X^{*} )$, $D( T_1 | X^{**}, T_2 | X^{**} )$, we do not suppress the degree-2 vertices
$v_1,v_2,u_1,u_2$. Note that, due to the cardinality constraints on $X^{*}$ and $X^{**}$, $p = tw(G^{*})$ and $q = tw(G^{**})$ because $D( T_1 | X^{*}, T_2 | X^{*} )$ can be
obtained from $G^*$ by suppressing the degree-2 vertices which does not alter the treewidth (because the pathological case of Observation \ref{obs:2} does not apply). Similarly for the other component.  Assume without loss of generality that $u_1$
and $u_2$ are in $G^{*}$, and $v_1$ and $v_2$ are in $G^{**}$. Let $\mathbb{T}^{*}$ and $\mathbb{T}^{**}$ be minimum-width tree decompositions of $G^{*}$ and $G^{**}$
respectively. Locate a bag $B^{*}$ of $\mathbb{T}^{*}$ that contains $u_1$ and a bag $B^{**}$ of $\mathbb{T}^{**}$ that contains $v_1$. Introduce a bag $\{u_1, v_1\}$ and
insert it between $B^{*}$ and $B^{**}$. Clearly, the width in this merged tree decomposition is not altered. It remains only to ensure that the decomposition covers the edge
$\{u_2, v_2\}$. This can be achieved simply by adding (say) $u_2$ to every bag in the tree decomposition of $G^{**}$, which increases the size of all bags by at most one. The
result follows.

Now, we deal with the case where $|X^*| \leq 2$ and/or $|X^{**}| \leq 2$. First of all, we observe that since $T_1, T_2$ are incompatible by assumption, it is not the case that
$|X^*|,|X^{**}| \leq 2$ at the same time.
So, at least one of $|X^*|,|X^{**}|$ must be at least 3. Suppose $|X^*| = 2$ and $|X^{**}| \geq 3$. Observe that in this case $tw(G^*) = 2 \neq p = 1$ but $tw(G^{**}) = q \geq 2$, so
$\max(p,q) \geq 2$. Hence the construction from the previous case - adding bag $\{u_1, v_1\}$ and then adding $u_2$ to all bags -  again cannot increase the width of the decomposition by
more than 1. The case $|X^{*}|=1$ is somewhat strange because then
$D( T_1 | X^{*}, T_2 | X^{*} )$ is just a single vertex. However, the upper bound still
goes through because $tw(G^{**}) = q \geq 2$ and $D(T_1, T_2)$ can be obtained from
$G^{**}$ by connecting the two roots of $G^{**}$ by an edge and then subdividing this
new edge with a single degree-2 vertex. Adding an edge to a graph can increase its treewidth
by at most 1, and edge subdivision is treewidth invariant.
\end{proof}

Now, let $[G^{*}]$ be the graph obtained from $G^{*}$ by adding the edge $\{u_1, u_2\}$, and $[G^{**}]$ be obtained from $G^{**}$ by adding the edge $\{v_1, v_2\}$. See again
Figure \ref{fig:clusterreduction}.

\begin{obs}
\label{obs:starsandwich} $tw(G^{*}) \leq tw([G^{*}]) \leq tw(D(T_1,T_2))$ and $tw(G^{**}) \leq tw([G^{**}]) \leq tw(D(T_1,T_2))$.
\end{obs}
\begin{proof}
The lower bounds are immediate by a standard minor argument. The upper bounds are also obtained via minors. Specifically, observe that $[G^{*}]$ can be obtained from $D =
D(T_1, T_2)$ by completely contracting the part of $D$ that lies between $v_1$ and $v_2$ (i.e. the $X^{**}$ part of $D$). A symmetrical argument holds for $[G^{**}]$ by
completely contracting the $X^{*}$ part of $D$.
\end{proof}

The following theorem strengthens Lemma \ref{lem:cr_bound} by adding necessary and sufficient conditions for the lower bound to hold.

\begin{thm}
\label{thm:cr_nec_suff} Consider Lemma \ref{lem:cr_bound}. Assume without loss of generality that $p \leq q$. Then $tw(D(T_1, T_2))  = \max(p,q)$ if and only if the following
holds:
\begin{enumerate}
\item (Case $p < q$): $tw( [G^{**}] ) = tw(G^{**})$,
\item (Case $p = q$):  $tw( [G^{**}] ) = tw( G^{**} )$ and $tw( [G^{*}] ) = tw( G^{*} )$.
\end{enumerate}
\end{thm}
\begin{proof} We consider both cases and both directions of implication.
\begin{enumerate}
\item (Case $p < q$, $\Rightarrow$) Assume $p < q$ and $tw(D(T_1, T_2))  = \max(p,q) = q$. Now, by  Observation \ref{obs:starsandwich},  $tw([G^{**}]) \leq tw(D(T_1,T_2)) = q
    = tw(G^{**})$. The bound $tw(G^{**}) \leq tw([G^{**}])$ also follows from Observation \ref{obs:starsandwich}, so $tw([G^{**}]) = tw(G^{**})$.

\item (Case $p = q$, $\Rightarrow$) Assume $p = q$ and $tw(D(T_1, T_2))  = \max(p,q) = p = q$. Both $tw( [G^{**}] ) = tw( G^{**} )$ and $tw( [G^{*}] ) = tw( G^{*} )$ follow
    from  Observation \ref{obs:starsandwich}.

\item (Case $p < q$, $\Leftarrow$) Observe that the statement $tw( [G^{**}] ) = tw(G^{**})$ holds if and only if there exists a minimum-width tree decomposition of $G^{**}$
    in which $v_1$ and $v_2$ are both in the same bag $B^{**}$. So, let us assume the existence of such a tree decomposition $\mathbb{T}^{**}$ and bag $B^{**}$. Construct a
    minimum-width tree decomposition $\mathbb{T}^{*}$ of $G^{*}$. Suppose $\mathbb{T}^{*}$ contains a bag $B^{*}$ that contains both $u_1$ and $u_2$. We can merge
    $\mathbb{T}^{*}$ and $\mathbb{T}^{**}$ by inserting bags $\{u_1, v_1, u_2\}$ and $\{u_2, v_1, v_2\}$  between $B^{*}$ and $B^{**}$. The size-3 bags do not influence the
    width of the decomposition, so $tw( D(T_1, T_2)) \leq \max(p,q)$, and $tw(D(T_1,T_2)) = \max(p,q)$ then follows from Lemma \ref{lem:cr_bound}. If no such bag $B^{*}$
    exists then create it by first adding (say) $u_2$ to every bag of $\mathbb{T}^{*}$. The addition of $u_2$ to every bag potentially increases the width of $\mathbb{T}^{*}$
    by 1, but due to the fact that $p < q$ we have $p+1 \leq q$, so $\max(p+1, q) \leq \max(p,q)$ and the earlier argument goes through.

\item (Case $p=q$, $\Leftarrow$) This is very similar to the (Case $p < q$, $\Leftarrow$) argument. The main difference is that, due to the strengthened starting assumption,
    both bags $B^{**}$ and $B^{*}$ are guaranteed to exist. Hence the ``If no such bag $B^{*}$...'' part of the argument will never be required.

\end{enumerate}
\end{proof}

The above results show that the treewidth of the display graph behaves rather differently around common splits than other phylogenetic incongruence measures. Many such measures
are (essentially) additive (i.e. the distance is the sum of the $X^{*}$ and $X^{**}$ parts) \cite{BSS06,linz2011cluster,bordewich2017fixed}, contrasting with the maximum function used in treewidth. \steven{As we demonstrate later in Section \ref{sec:gapBig} this is one of the reasons why treewidth distance can be substantially lower than, for example, $d_{MAF}$.} A second point worth noting is that, while Theorem \ref{thm:cr_nec_suff} describes necessary and
sufficient conditions for the treewidth of the display graph to achieve the lower bound, it is not yet clear what (phylogenetic) properties of $T_1$ and $T_2$ actually create
these conditions. Expressed differently, and for simplicity focussing on the case $p < q$: what properties do $T_1$ and $T_2$ need to have to ensure $tw( [G^{**}] ) =
tw(G^{**})$? \steven{It is perhaps relevant to observe that the graphs $[G^{*}], [G^{**}]$ can
themselves be viewed, modulo a treewidth-invariant suppression of a single degree-2 vertex, as display graphs of appropriately rooted phylogenetic trees. Taking $[G^{*}]$ as an example: take the two
trees $T_1|X^{*}$ and $T_2|X^{*}$ and attach a new placeholder taxon $\rho$ at points $u_1$
and $u_2$, respectively.}

\section{Diameters on $d_{tw}$}
\label{sec:upperbounds}

In this section we explore the question of how large can the treewidth of the display graph of two unrooted binary phylogenetic trees, both on $X$, can get. More precisely, we
consider the diameter $\Delta_n(d_{tw})$  defined as the maximum value $d_{tw}$ taken over all pairs of phylogenetic trees with $n$ taxa.
 Somewhat surprisingly,  we show that $\Delta_n(d_{tw})$ is bounded below and above by linear functions on $n$.
  To prove this, we first present a general result showing how we can embed an arbitrary graph into display
graphs (as minors) without adding too many extra edges or vertices.

\begin{thm}
\label{thm:embed} Let $G=(V,E)$ be an undirected (multi)graph with $n$ vertices and maximum degree $d \geq 2$. Then we can construct two unrooted binary phylogenetic trees
$T_1$  and $T_2$ such that both trees have $O(nd)$ taxa, $O(nd)$ nodes and $O(nd)$ edges (and hence their display graph has $O(nd)$ nodes and edges) and $G$ is a minor of
$D(T_1, T_2)$.
\end{thm}
\begin{proof}
The construction can easily be computed in polynomial time. We start by selecting an arbitrary unrooted binary tree $T$ on $n+2$ taxa. Set $T_1 := T$ and $T_2 := T$. The idea
is that the $n$ internal nodes of $T_1$ are in bijection with the $n$ vertices of $G$. We will add the edges of $G$ one at a time, in the following manner. If an edge $e =
\{u,v\}$ of $G$ already exists within $T_1$, the edge is already encoded so there is nothing to do. If not, we subdivide an arbitrary edge in $T_2$ and let $y$ be the
subdivision node. We then introduce two new taxa $x^{e}_1$ and $x^{e}_2$ and a new vertex $z$ in $T_2$, and add the following edges: $\{u, x^{e}_1\}$, $\{x^{e}_1,  z\}$, $\{z,
y\}$, $\{z, x^{e}_2\}$ and $\{x^{e}_2, v\}$. The first and last of these edges is in $T_1$, the rest are in $T_2$. In the display graph the path $u, x^{e}_1, z, x^{e}_2, v$
will become the image of the edge $\{u,v\}$ (in the embedding of the minor). After encoding all the edges, $T_1$ and $T_2$ will each have at most $k  = (n+2)+2|E|$ taxa, so
(because $T_2$ remains binary) each will have at at most $k-2$ internal nodes and each at most $2k-3$ edges. Now, observe that the $n$ internal nodes of $T_1$ might have degree
as large as $d+3$. To turn $T_1$ into a binary tree we replace each vertex $u$, where $deg(u) > 3$, by a path of $t =deg(u)-2$ vertices $u_1, ..., u_t$. The first two edges
incident to $u$ are now made incident to $u_1$, the final two edges incident to $u$ are made incident to $u_t$, and each of the remaining edges is made incident to exactly one
of the nodes $u_2, ..., u_{t-1}$. (When obtaining $u$ from the embedding of $G$, the idea is that the edges of the path will be contracted to retrieve $u$). This transformation
does not alter the number of taxa, so $T_1$ and $T_2$ now have both the same number of internal nodes and edges (i.e. at most $k-2$ and $2k-3$ respectively). Due to the fact
that $G$ has maximum degree $d$, $|E| \leq nd/2$. We conclude that both trees each \twu{has} at most $(n+2) + nd$ taxa, at most $n(d+1)$ internal nodes and at most $2n + 4 + 4|E| -
3 \leq 2n + 1 + 2nd$ edges. It follows that $D(T_1, T_2)$ has at most $2n(d+1) + ((n+2)+nd)$ nodes in total and at most $4n+2+4nd$ edges.
\end{proof}

Applying the last theorem on complete graphs leads to a lower bound on $\Delta_n(d_{tw})$ which grows linearly on $\sqrt{n}$.  To get a better lower bound,   below we use the
fact that there are cubic expanders on $n$ vertices with treewidth at least $\epsilon n$, for some constant $\epsilon > 0$ (see, for example,
\cite{DBLP:journals/jct/GroheM09,DBLP:conf/gd/DujmovicEW15}).

\begin{cor}
\label{cor:lintw}
We have $\Delta_n(d_{tw})=\Theta(n)$ for $n\geq 4$, that is, $\Delta_n(d_{tw})$ is bounded below by a linear function on $n$ and above by a linear function on $n$.
\end{cor}
\begin{proof}
As mentioned above, it is a well known fact that there are cubic expanders on $n$ vertices with treewidth linear in $n$. The maximum degree of such graphs is 3, so for every
cubic expander $G$ there exists (by the previous theorem) a display graph on $O(n)$ vertices that contains $G$ as a minor. Hence, the treewidth of the display graph is at least
that of the expander, which is linear in $n$. This establish the linear lower bound on $\Delta_n(d_{tw})$. The upper bound follows from
$$
\Delta_n(d_{tw}) \leq \Delta_n(d_{TBR}) \leq n-3 -\left \lfloor \frac{\sqrt{n-2}-1}{2} \right\rfloor ,
$$
where the second inequality follows from~\cite[Theorem 1.1]{ding2011agreement}.

\end{proof}

\begin{figure}[ht]
\centering
\includegraphics[scale=0.95]{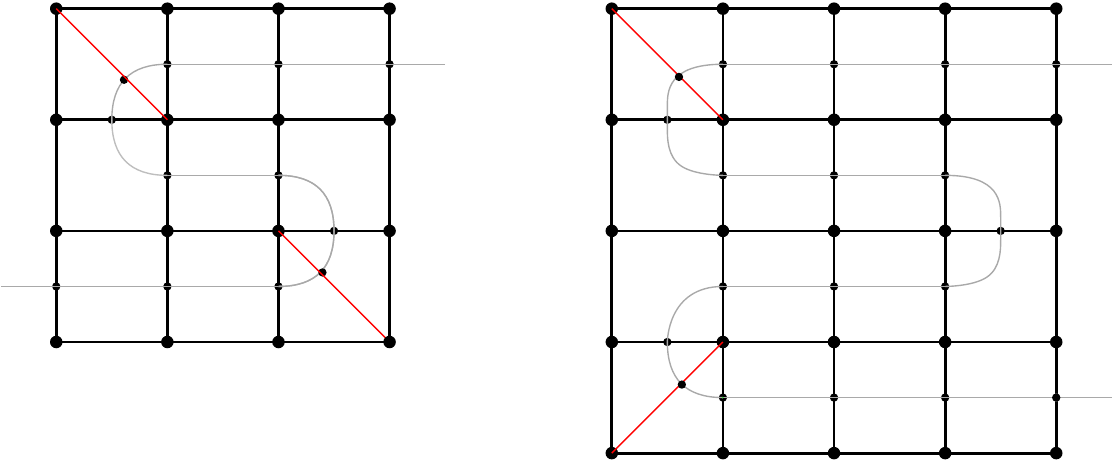}
\caption{Embedding grid minors in the display graphs of two unrooted binary trees: grids with even side length (left; $k=4$) and odd side length (right; $k=5$). Taxa are shown as small dots inside the grid. Both trees have exactly $(k-1)^2 + 3$ taxa.}
\label{fig:gridpack}
\end{figure}

The construction (and bounds) described in Theorem \ref{thm:embed} can be refined significantly in specific cases. Consider the $k \times k$ grid graph, which has maximum
degree 4 and $k^2$ nodes. When taking $n=k^2$ the theorem yields a bound of $ \approx 13n$ nodes. However, consider the construction shown in Figure \ref{fig:gridpack}, which distinguishes
the cases $k$ even and $k$ odd. The two sides of the curve indicate the two trees that are needed and the points at which the curve touches the grid become the taxa of the two
trees. Note that the \steven{red} edges are added simply to ensure that all 4 corners of the grid (which have degree-2) can be correctly retrieved when taking minors; 2 of the
corners are already present (because the curve intersects a neighbouring edge) but the other 2 require the addition of the red edges\footnote{Note that if we ``round off''
the 4 corners of the grid its treewidth (which is $k$) is unaffected and the red edges are not required.}. As in the theorem the degree-4 nodes can be split into two
degree-3 nodes. In both the odd and even cases it can be verified that both the resulting unrooted binary trees have $(k-1)^2 + 3$ taxa and thus that the display graph has
$3(k-1)^2 + 5$ nodes in total. This is $\approx 3n$, a significant improvement on the generic bound. In fact it is not far from ``best possible''. A $k \times k$ grid contains
$(k-1)^2$ chordless 4-cycles, and because a tree cannot contain a cycle the embedding of each cycle must pass through at least 2 taxa in the display graph. Each taxon can be
shared by at most two 4-cycles (because the display graph has maximum degree 3) yielding a lower bound on the number of taxa required of $(k-1)^2$.

\section{Display graphs formed from trees and networks}
\label{sec:treeNets}

In this section we will consider the display graph formed by an unrooted binary phylogenetic network $N = (V,E)$ and an unrooted binary phylogenetic tree $T$ both on the same set of taxa $X$. We will show
upper bounds on the treewidth of $D(N,T)$ in term of the reticulation number $r(N) (= |E| - (|V|-1))$ and the treewidth $tw(N)$ of $N$.



We begin with the first result that gives a sharp upper bound on the treewidth of the display graph in terms of the reticulation number of $N$.

\begin{lem}
\label{lem:reticbound} Let $N = (V,E)$ be an unrooted binary phylogenetic network and $T$ an unrooted binary phylogenetic tree, both on $X$, where $|X| \geq 3$. If $N$ displays
$T$ then $tw( D(N,T)) \leq r(N)+2$. \label{lem:twbound}
\end{lem}
\begin{proof}
\steven{Due to the fact that $N$ displays $T$, there is a subgraph $T'$ of $N$ that is a subdivision of $T$.  If $T'$ is a spanning tree of $N$, then let $N' = T'$. Otherwise,
construct a spanning tree $N'$ of $N$ by greedily adding edges to $T'$ until all vertices of $N$ are spanned.} At this point, $N'$ contains exactly $|V|-1$ edges and consists of a subdivision of $T$
from which possibly some unlabelled pendant subtrees (i.e. pendant subtrees without taxa) are hanging. We argue that $D(N',T)$ has treewidth 2, as follows. First, note that $D(T,T)$ has treewidth 2, because $T$ is trivally compatible with $T$ (and $|X| \geq 3$). Now, $D(T,T)$ can be obtained from $D(N', T)$ by repeatedly deleting unlabelled vertices of
degree 1 and suppressing unlabelled degree 2 vertices. These operations cannot increase or decrease the treewidth (because $tw(D(N',T)) = 2$ and because the pathological case of Observation \ref{obs:2} does not apply here).
 Hence, $D(N', T)$ has treewidth 2. Now, $D(N,T)$ can be obtained from $D(N',T)$ by adding back the $r(N) =
|E|-(|V|-1)$ missing edges. The addition of each edge can increase the treewidth by at most 1, so $tw(D(N,T))
\leq 2 + r(N)$.
\end{proof}

We note that this bound is sharp, since if $N=T$ then $r(N)=0$ and $D(N,T)$ has treewidth 2. Also, observe that an essentially unchanged argument shows that if two networks
$N_1$ and $N_2$, both on $X$, both display some tree $T$ on $X$, then $tw(D(N_1, N_2)) \leq r(N_1) + r(N_2) + 2$.

Now we derive a second upper bound in terms of the treewidth of the network $N$.

\begin{lem}
\label{lem:twbound} Let $N = (V,E)$ be an unrooted binary phylogenetic network and $T$ an unrooted binary phylogenetic tree, both on $X$, where $|X| \geq 3$. If $N$ displays
$T$ then $tw( D(N,T)) \leq 2 tw(N) + 1$.
\end{lem}
\begin{proof}
\steven{Since $N$ displays $T$ there is a subgraph $N'$ of $N$ that is a subdivision of $T$. We consider the surjection function $f$ defined in the preliminaries (which maps from vertices of $N'$ to vertices of $T$). Informally, $f$ maps taxa to taxa and  degree-3 vertices of $N'$ to the corresponding vertex of $T$. Each degree-2 vertex of $N'$ lies on a path corresponding to an edge $\{u,v\}$ of $T$; such vertices are mapped to $u$ or $v$, depending on how exactly the surjection was constructed.}


Now, consider any tree decomposition $t$ of $N$. Let $k$ be the width of the tree decomposition i.e., the largest bag in the tree decomposition has size $k+1$. We will
construct a new tree decomposition $t'$ for $D(N,T)$ as follows. For each vertex
$u' \in V(N')$
we add $f(u')$ to every bag that contains $u'$.
To show that $t'$ is a valid tree decomposition for $D(N,T)$ we will show
that it satisfies all the treewidth conditions. Condition (tw1) holds because $f$ is a surjection.
For property (tw2) we need to show that for every edge $e = \{ u,v \} \in E(T)$, there exists some bag $B \in V(t):
\{ u,v \} \subset B$. For this we use the third property of $f$ described in our observation: $\forall \{ u,v \} \in E(T), \exists_1 \{ \alpha, \beta \} \in E(N'):$ $f(\alpha)
= u$ and $f(\beta) = v$. For each $e = \{u,v\} \in E(T)$, let $\{ \alpha, \beta \} \in E(N')$ be
the edge which is
mapped through $f$ to e. Since $\{ \alpha, \beta \} \in E(N)$, there must be a bag $B \in V(t)$ that contains both of $\alpha, \beta$. Since $f(\alpha) = u$ and $f(\beta) = v$,
both of $u,v$ will be added into $B$. For the last property (tw3) we need to show that the bags of $t$ where $u \in V(T)$ have been added form a connected component. For this,
we use property (2) of the function $f$: $\forall v \in V(T)$, the set $\{ u \in V(N'): f(u) = v \}$ forms a connected subtree in $N'$. \steven{Hence, the set of bags that contain at least one element from
$\{ u \in V(N'): f(u) = v \}$ form a connected subtree in the tree decomposition.} These are the bags to which $v$ is
added, ensuring that (tw 3) indeed holds for $v$.

We now calculate the width of $t'$: Observe that the size of each bag can at most double. This can happen when every vertex in the bag is in $V(T')$, and $f(u') \neq f(v')$ for
every two vertices $u', v'$ in the bag. This causes the largest bag after this operation to have size  at most $2(k+1)$ i.e., the width of the new decomposition is at most
$2k+1$.
\end{proof}

Combining these two lemmas yields the following:

\begin{cor}
\label{cor:twbound}
Let $N$ be an unrooted binary phylogenetic network and $T$ be an unrooted binary phylogenetic tree, both on $X$. Then if $N$ displays $T$, $tw( D(N,T)) \leq \min\{ 2 tw(N) + 1, r(N) + 2 \}$.
\end{cor}

\begin{figure}[ht]
\centering
\includegraphics[scale=0.8]{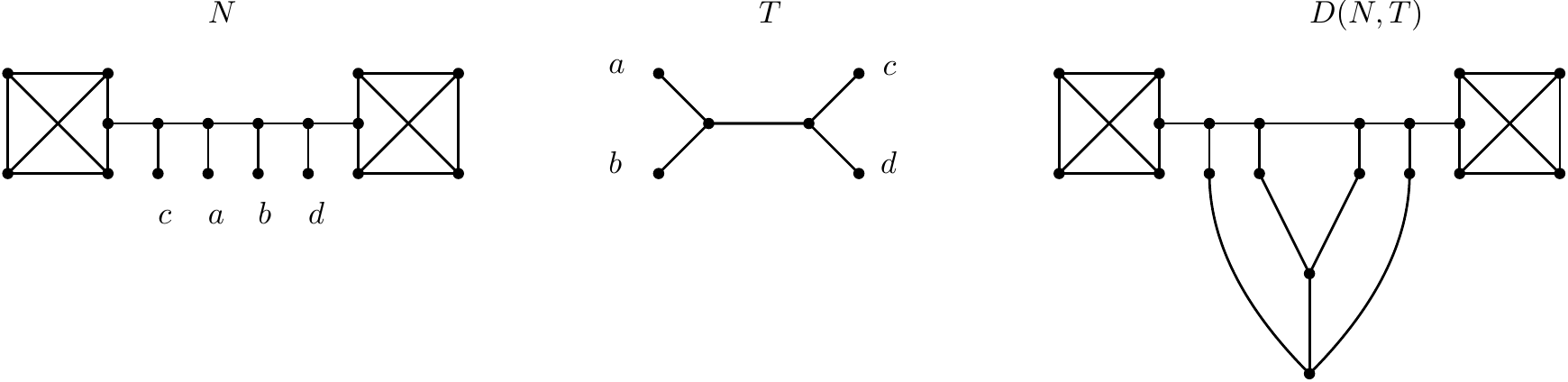}
\caption{The network $N$ does not display the tree $T$ but the treewidth of their display graph is equal to the treewidth of $N$ which is equal to 3.}
\label{fig:net}
\end{figure}

The above results raises a number interesting points. First, consider the case where a binary network $N$ \emph{does not} display a given binary phylogenetic network $T$. As we can see in Figure \ref{fig:net}, there
is a network $N$ and a tree $T$ such that $N$ does not display $T$ and yet the treewidth of their display graph is equal to the treewidth of $N$ which (as can be easily verified) is equal to three. Hence ``does not display'' does not necessarily
cause an increase in the treewidth. On the other hand, the results in Section \ref{sec:upperbounds} show that for two incompatible unrooted binary phylogenetic trees (vacuously: neither of which displays the other, and both of which have treewidth 1) the treewidth of the display graph can be as large as linear in the size of the trees. The increase
in treewidth in this situation is asymptotically maximal. So the relationship between ``does not display'' and treewidth
is rather complex. Contrast this with the
bounded growth in treewidth articulated in Corollary \ref{cor:twbound}. Such bounded growth opens the door to
algorithmic applications. In Section \ref{sec:courcelle} (in the Appendix) we leverage this bounded growth to
obtain a (theoretical) FPT result for determining whether a network displays a tree.

Second, observe that $tw(N)$ is equal to the maximum treewidth ranging over all biconnected components of $N$. An upper bound on the treewidth of a biconnected component $K$ is
$r(K)$ i.e. the number of edges that need to be deleted from $K$ to obtain a spanning tree of the component. In phylogenetics the maximum value of $r(K)$ ranging over all
biconnected components $K$ of $N$ is a well-studied parameter known as \emph{level} \cite{GBP2012,HusonRuppScornavacca10}. So $tw(N) \leq \emph {level}(N)$. Hence, if $N$
displays $T$, then the treewidth of $D(N,T)$ is also bounded as a function of the level of $N$.

\section{The unit ball of $d_{tw}$ compared to that of $d_{TBR}$ and $d_{MP}$}
\label{sec:unitball}

\twu{In this section we will compare the unit ball neighborhood of $d_{tw}$ with those of $d_{TBR}$ and $d_{MP}$. Recall that given a distance $d$ and a phylogenetic tree $T$ on $X$ the unit
neighborhood of $T$ under $d$ is the set of all phylogentic trees $T'$ on $X$ with the property that $d(T,T') = 1$ (see, e.g.\cite{humphries2013neighborhoods, moulton2015parsimony}, for results that characterise the unit ball neighbourhoods of $d_{TBR}$ and $d_{MP}$). } We will begin by comparing treewidth with Maximum Parsimony
(MP) unit neighborhoods.


\begin{thm}
\label{thm:mp} Suppose that $T$ and $T'$ are a pair of  unrooted binary phylogenetic trees on $X$ with $d_{MP}(T,T')=1$ or $d_{TBR}(T,T') = 1$. Then we also have  $d_{tw}(T,T')=1$.
\end{thm}

\begin{figure}[ht]
\centering
\includegraphics[scale=0.8]{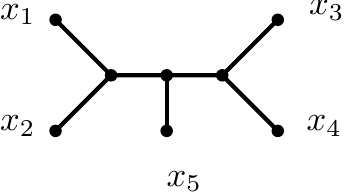}
\caption{\steven{Without loss of generality we can assume that an unrooted binary phylogenetic tree on 5 taxa has (up to relabelling of taxa) the topology $T_1$.}}
\label{fig:five:trees}
\end{figure}

\begin{proof}
\steven{Note that, because both TBR and MP distance are metrics (and thus satisfy the identity of indiscernibles property) we can assume that $T_1$ and $T_2$ are incompatible.}
We will first show that the  claim is true for the TBR distance. Take two (necessarily incompatible) binary phylogenetic trees $T,T'$ such that $d_{TBR}(T,T') = 1$. By
combining the results of \cite{AllenSteel2001} where it was shown that $d_{TBR}(T_1,T_2) = d_{MAF}(T_1,T_2)+1$ and the result of \cite{kelk2015} where it was shown that
$tw(D(T_1,T_2))  \leq d_{MAF}(T_1,T_2) + 1$ we have that
$$
tw(D(T_1,T_2))  \leq d_{TBR}(T_1,T_2) + 2.
$$

Now if $T,T'$ are such that $d_{TBR}(T,T') = 1$ we conclude by the above that $tw(D(T,T')) \leq 3$ and by the assumption that $T,T'$ are incompatible we have that $d_{tw}(T,T') =
1$.

Now we will deal with the Maximum Parsimony distance. Let $T,T'$ be two (necessarily incompatible) unrooted binary phylogenetic trees such that $d_{MP}(T,T') = 1$.  Using Theorem~\ref{thm:subtree:reduction}, we
assume without any loss of generality that $T$ and $T'$ share no common pendant subtrees. Therefore, we can apply ~\cite[Theorem 6.4]{moulton2015parsimony} on $T,T'$ which characterizes the
unit ball neighborhood of the maximum parsimony distance. There it was shown that $d_{MP}(T,T') = 1$ if and only if either (1) $d_{TBR}(T,T') = 1$, in which case we are done
since we are in the TBR case or (2)  $d_{TBR}(T,T') = 2$ and \twu{using common pendant subtree (CPS) reductions we can transform $T$ and $T'$ into a pair of trees with \textit{precisely} five taxa.} \steven{(All unrooted binary phylogenetic trees on 5 taxa are caterpillars and modulo relabelling of taxa there is only one caterpillar topology on 5 taxa.)} \twu{Since $d_{tw}$ is preserved by CPS reduction in view of Theorem~\ref{thm:subtree:reduction}, we
can  assume without loss of generality that $T$ and $T'$ both have 5 taxa, and $T$ is the tree $T_1$ depicted in Fig.~\ref{fig:five:trees}.} Let $D = D(T,T')$ be the display graph formed from $T$ and $T'$ in which we subsequently suppress
all vertices of degree-2. (Suppression does not alter the treewidth, by Observation \ref{obs:2b}.) It is easy to observe that $D$ has at most \steven{(in fact, exactly)} 6  vertices.

Now, assume that $tw(D) > 3$ so that $d_{tw}(T,T') > 1$. Then, $D$ must have as a minor one of the forbidden minors for treewidth 3. In other words, one of the forbidden minors
for treewidth 3 can be obtained by a series of edge deletions/contractions on $D$.  There are precisely 4 forbidden minors for treewidth 3 \cite{DBLP:journals/dm/ArnborgPC90},
2 of which are on 6 vertices or less: the $K_5$ and the Octahedron graph. Both of them have uniform degree 4. On the other hand, recall that the degree of each vertex of $D$ is
3 (because $T,T'$ are unrooted \emph{binary} phylogenetic trees), so each degree-4 vertex of the minor maps to at least 2 vertices of $D$. This is clearly impossible.  So $D$
cannot contain as a minor any of the forbidden minors for treewidth 3 which shows that $tw(D) \leq 3$. By assumption, $T,T'$ are incompatible so $tw(D(T,T')) = 3 \Rightarrow
d_{tw}(T,T') = 1$.
\end{proof}

In the following section we will show that the converse of the above claim, namely that $d_{tw}(T_1,T_2) = 1 \Rightarrow d_{MP}(T_1,T_2) = 1$ is certainly not true (and that the same holds for the TBR distance.)

\section{On the gap between $d_{tw}$ and $d_{TBR}, d_{MP}$}
\label{sec:gapBig}

The purpose of this section is to explore how far treewidth distance $d_{tw}$ can be from the other two distances considered in this manuscript, namely maximum parsimony
distance $d_{MP}$ and TBR distance $d_{TBR}$. In particular we will provide an example of a sequence of \steven{pairs of} trees whose treewidth distance is as low as 1 (i.e.,
the treewidth of their display graph is at most 3) but such that the corresponding TBR and MP distances can be arbitrarily large.

\begin{figure}[ht]
\centering
\includegraphics[scale=1]{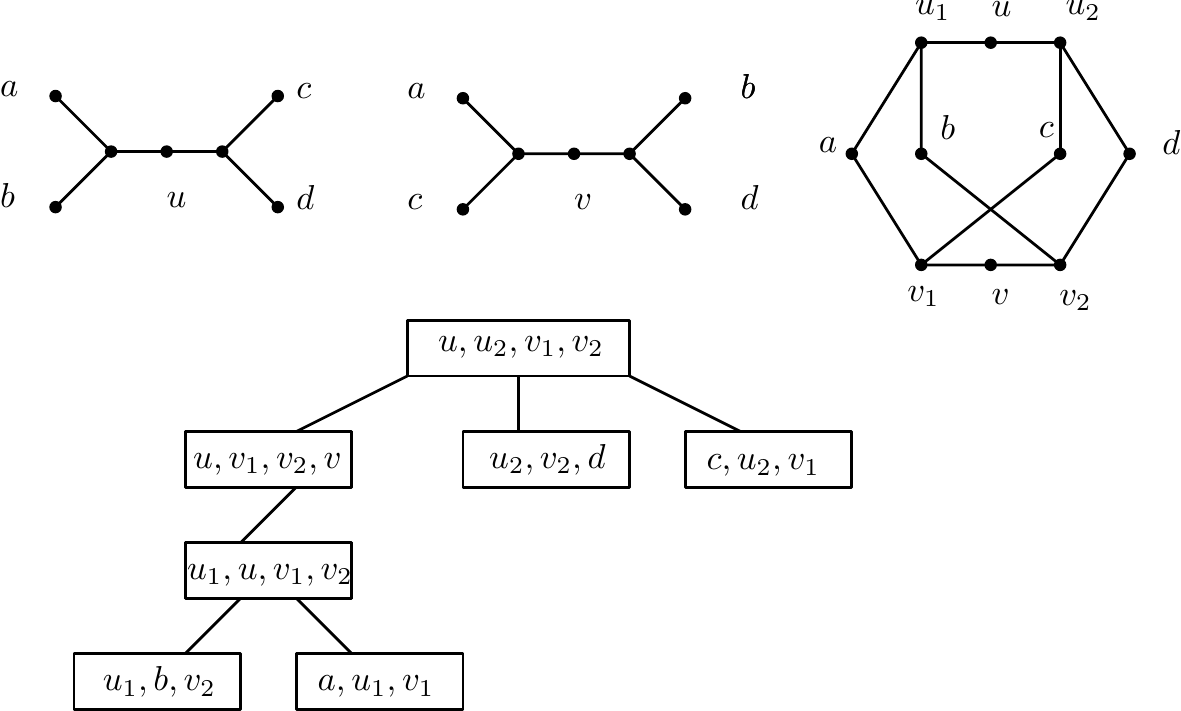}
\caption{Top: The two quartets $ab|cd$ and $ac|bd$ and their corresponding display graph (denoted $D^{0} = D$ in
the proof of Claim \ref{claim:always3}). Bottom: a width-3 tree decomposition of $D$ in which $u,v$ are in the same bag.}
\label{fig:quartets}
\end{figure}

\noindent
The construction starts with the 2 incompatible \emph{quartets} (unrooted binary trees on 4 taxa) $T_1 = ab|cd$ and $T_2 = ac|bd$. Without any loss of generality, we assume that both of the quartets contain a degree-2 vertex in the
``middle" namely, vertices $u,v$ respectively. (See Figure \ref{fig:quartets}). \steven{Note that with or without these degree-2 vertices the display graph has treewidth exactly 3 (by Observation \ref{obs:2b})}.


Given a tree $T$ with a single degree-2 vertex we define the following \textit{doubling} operation as follows:

\begin{description}
\item[Doubling tree operation:] Given a tree $T$, with a unique degree-2 vertex $v$, the doubling of $T$, denoted by $(T,T)$, is constructed as follows: we take 2 copies of
    $T$ and we join with an edge their unique degree-2 vertices. We subdivide this new edge such that $(T,T)$ has a unique degree-2 vertex.
\end{description}

\begin{figure}[ht]
\centering
\includegraphics[scale=1]{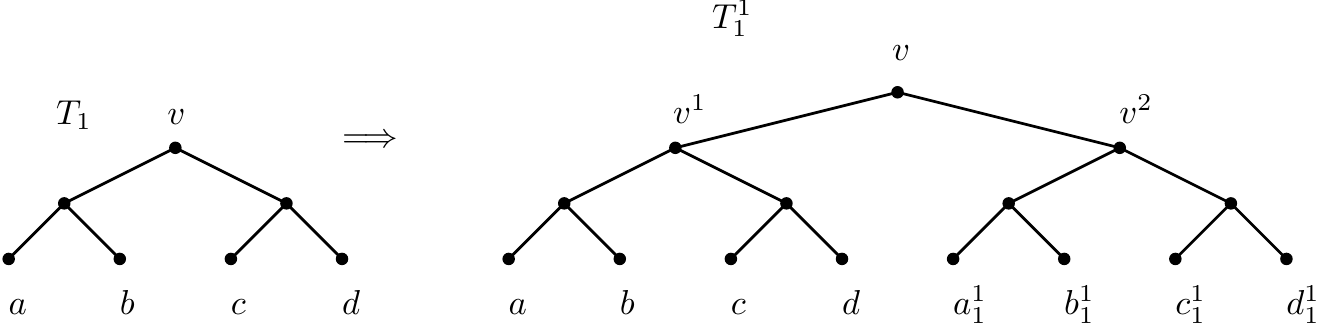}
\caption{An example of doubling the tree $T_1 = ab|cd$. When we create the second copy we label the taxa of the new copy appropriately to reflect the stage of doubling they
appear at (superscript) and the tree to which they belong (subscript).}
\label{fig:doubling}
\end{figure}

This operation will be the base of our construction. We will construct trees $T^i_1$ and $T^i_2$, for any step $i$, inductively as follows: $T^1_1 = (T_1,T_1)$ and $T^{i+1}_1 =
(T^i_1,T^i_1)$. Similarly for $T^2_2$ and subsequently for $T^{i+1}_2 = (T^i_2,T^i_2)$. Let $D^i$ be the display graph of $T^i_1$ and $T^i_2$. Observe that since we start from
$T_1, T_2$ on a common set of 4 taxa $\{a,b,c,d\}$, all the new doubled trees are on the same taxon set by labelling the new leaves appropriately, and so their display graph is
well defined and unique. Initially, let $D = D^0$ be the display graph of $T^0_1 = T_1$ and $T^0_2 = T_2$. We will show that $tw(D^i) = 3, \forall i$.


\begin{claim}
For every step $i$ we have that $d_{tw}(T^i_1,T^i_2) = 1$. Equivalently, we have that $tw(D^i) = 3$.
\label{claim:always3}
\end{claim}

\begin{proof}
The proof is by an inductive argument. For the base of the induction, we first construct a tree decomposition of width 3 with specific properties: see Figure \ref{fig:quartets}.

As is apparent from the base case, we can assume without any loss of generality that the two degree-2 vertices $u,v$ in $T_1,T_2$ respectively, are in the same bag of the tree
decomposition of their display graph $D$. We will exploit this fact in the following. For the induction step we assume that the display graph $D^i$ formed by $T^i_1$ and
$T^i_2$ has treewidth 3. We will show a tree decomposition for $D^{i+1}$ of width equal to the width of the tree decomposition of $D^i$. We
can construct $D^{i+1}$ from $D^i$ as follows: take two copies of $D_i$, let's call them $D^i_1$ and $D^i_2$. Each copy $D^i_j, j \in \{1,2\}$ has two degree-2
vertices: one, let's call it $u^i_j$  is the degree-2 vertex resulting after repeated doubling of the $T_1$ tree and the other, let's call it $v^i_j$ from doubling the $T_2$
tree. For each display graph $D^i_j$ let $\mathbb{T}^i_j$ be its tree decomposition which by the inductive hypothesis has width 3. Moreover, as explained, we can assume without any
loss of generality that the two degree-2 vertices $u^i_j$ and $v^i_j$ are in the same bag $B_j$.
Observe that $D^{i+1}$
has two new degree two vertices, $u^*, v^*$: $u^*$ will be connected with each $u^i_j$ and $v^*$ with each $v^i_j$, $j \in \{1,2\}$. Construct $\mathbb{T}^{i+1}$ as follows:
locate the bags $B_j$ that contain $\{u^i_j, v^i_j \}$, $j \in \{1,2\}$. Such bags exist by the inductive hypothesis. Create the following chain of bags: $B_1 - \{ u^*, u^i_1, v^i_1
\} - \{ u^*, v^*, v^i_1 \} - \{ u^*,v^*,v^i_2 \} - \{ u^*, u^i_2, v^i_2 \} - B_2$. It is immediate that $\mathbb{T}^{i+1}$ is a valid tree decomposition for $D^{i+1}$ of width
no higher than the width of $\mathbb{T}^i$ (and $u^*, v^*$ are in the same bag) so the claim follows.
\end{proof}

So the treewidth distance $d_{tw}$ of $T^i_1$ and $T^i_2$ remains $1$ for any $i$. We will now give lower bounds on $d_{TBR}(T^i_1,T^i_2)$.
%
We claim that $d_{TBR}(T^i_1, T^i_2) > d_{TBR}(T^j_1,T^j_2)$ for $i > j$. In particular $d_{TBR}(T^{i+1}_1, T^{i+1}_2) >
d_{TBR}(T^{i}_1,T^i_2)$, for all $i \geq 0$. We will prove the claim using the \emph{maximum agreement forest} distance which, by the result of Allen and Steel \cite{AllenSteel2001},
is equivalent to TBR: $d_{MAF}(T_1,T_2) = d_{TBR}(T_1,T_2) +1$. \steven{(See the preliminaries for definitions pertaining to agreement forests).}
First of all, it is not too difficult to verify that (after suppression of the two degree-2 vertices\footnote{Agreement forests are unaffected by suppression of degree-2 vertices.}) $d_{MAF}(T_1,T_2) = 2$.

Let $T^{i+1}_j$ be the two trees obtained after we double $T^i_j$, for $j \in \{1,2\}$ and let $d_{MAF}(T^i_1, T^i_2) = p \in \mathbb{N}^+$. We assume without loss of generality that neither of $T^{i+1}_1, T^{i+1}_2$ has a degree-2 vertex. We distinguish between two cases: Let $e_1 (e_2)$ be the edge used to connect the two
copies of $T^i_1 (T^i_2)$ to construct $T^{i+1}_1 (T^{i+1}_2)$. We say that an edge is \emph{deleted} by an agreement forest if it is an edge that is deleted in order to obtain the agreement forest. It is easy to observe that if $e_1$ is deleted in an agreement forest, then so is $e_2$ because of the symmetric properties of the constructed graphs $T^{i+1}_1, T^{i+1}_2$.
Now, fix $m$ to be an arbitrary maximum agreement forest.

\begin{description}
\item[Edges $e_1 (e_2)$ are deleted by $m$:]
Note that by deleting $e_1 (e_2)$ we obtain two disjoint copies of the trees $T^i_1 (T^i_2)$.
In this case we observe that $d_{MAF}(T^{i+1}_1,
    T^{i+1}_2) = 2 d_{MAF}(T^i_1, T^i_2) = 2p$ since any maximum agreement forest that does not use $e_1 (e_2)$
can and should select a maximum agreement forest for the pair of trees $T^i_1, T^i_2$, and do this twice (since there
are two disjoint copies of these trees).
\item[Neither of these edges is deleted by $m$:] Then these edges are used by the image of some component $C$ of the agreement forest $m$. If we split $C$ into two pieces (at the edges $e_1$ and $e_2$) we increase the size of the agreement forest by 1 and obtain an agreement forest that does not use either edge $e_1$ or $e_2$. From the
previous case we know that any agreement forest that does not use these edges has at least $2p$ components.
Hence, $d_{MAF}(T^{i+1}_1, T^{i+1}_2) \geq 2p - 1$.



\end{description}

\begin{lem}
The MAF
distance between $T^{i+1}_1$ and $T^{i+1}_2$ is at least $2 \times d_{MAF}(T^i_1, T^i_2) -1 > d_{MAF}(T^i_1, T^i_2)$.
\end{lem}

\begin{thm}
\label{thm:tbrgrows}
There is at least one infinite subfamility of trees $T_1, T_2$ such that $d_{tw}(T_1,T_2) = 1$ whereas $d_{TBR}(T_1,T_2)$ is unbounded.
\end{thm}

Finally, we turn to $d_{MP}$:

\begin{thm}
\label{thm:tbrgrows}
There is at least one infinite subfamility of trees $T_1, T_2$ such that $d_{tw}(T_1,T_2) = 1$ whereas $d_{MP}(T_1,T_2)$ is unbounded.
\end{thm}
\begin{proof}
In fact, this is a strengthening of the previous theorem because $d_{MP}$ is always a lower bound on $d_{TBR}$. However, as $d_{MP}$ is less well-known we only sketch the construction. Observe that the tree $T^{i}_1$ contains $2^i$ copies of each taxon. We assign all the copies of taxa $a$ and $b$ the state 0, and all copies of taxa $c$ and $d$ the state
1. It can be easily verified (by applying e.g. Fitch's algorithm) that the parsimony score of $T^{i}_1$ on such a character
is at most $2^i$. However, on the same character the parsimony score of $T^{i}_2$ will be at least $2 \cdot 2^i$.
Hence, $d_{MP}(T^i_1, T^i_2) \geq 2\cdot 2^{i} - 2^{i}$ and this grows to infinity.
\end{proof}

\section{Discussion and open problems}
In this paper we presented several algorithmic and combinatorial results on the treewidth distance $d_{tw}$, including its behaviour under three commonly used tree reduction rules and its diameter and unit ball neighbourhood. There are a number of interesting problems remain open, and we discuss some of  them below.

A major open question is whether it is \textbf{NP}-hard to compute the treewidth distance $d_{tw}$ between two trees.  This is equivalent to compute the treewidth of the
display graph of these two trees, which is a cubic graph after suppressing all degree-2 vertices. Although computing the treewidth of general graphs is \textbf{NP}-hard, even
for graphs whose maximum degree is at most 9~\cite{arnborg1987complexity,bodlaender1997treewidth}, it is still unknown whether the treewidth of cubic graphs  can be computed in
polynomial time. Hence it is also interesting to understand the complexity of computing the treewidth of cubic graphs, and whether it has the same complexity of computing that
of display graphs. Moreover, irrespective of whether it is an \textbf{NP}-hard problem, it is of interest to explore whether the structure of display graphs can be leveraged to
compute their treewidth quickly \emph{in practice}.

Another question concerns the common chain reduction, that is, whether there exists a universal constant $d$ such that reducing common chains to length $d$, preserves the
treewidth  of the display graph? This is likely to require deep insights into forbidden minors - in particular the way they interact with chain-like regions of graphs (that are
not separators).

Initial numerical experiments suggest that treewidth distance can be ``low'' compared to
traditional phylogenetic distances, such as the well-known TBR distance. Is this phenomenon more widespread? In how far is this an artefact of the way treewidth distance decomposes around common splits?  Are there traditional phylogenetic distances and measures which are verifiably (and/or empirically) close to treewidth distance - and, if so, why? Finally,  could we leverage low treewidth distance to develop efficient algorithms (based on
dynamic programming over tree decompositions) for other phylogenetic distances and measures?

\bibliographystyle{plain}
\bibliography{bibliographyTOP}

\begin{thebibliography}{10}

\bibitem{AllenSteel2001}
B.~Allen and M.~Steel.
\newblock Subtree transfer operations and their induced metrics on evolutionary
  trees.
\newblock {\em Annals of Combinatorics}, 5:1--15, 2001.

\bibitem{Arnborg91}
S.~Arnborg, J.~Lagergren, and D.~Seese.
\newblock Easy problems for tree-decomposable graphs.
\newblock {\em Journal of Algorithms}, 12:308 -- 340, 1991.

\bibitem{DBLP:journals/dm/ArnborgPC90}
S.~Arnborg, A.~Proskurowski, and D.~G. Corneil.
\newblock Forbidden minors characterization of partial 3-trees.
\newblock {\em Discrete Mathematics}, 80(1):1--19, 1990.

\bibitem{arnborg1987complexity}
Stefan Arnborg, Derek~G Corneil, and Andrzej Proskurowski.
\newblock Complexity of finding embeddings in ak-tree.
\newblock {\em SIAM Journal on Algebraic Discrete Methods}, 8(2):277--284,
  1987.

\bibitem{BSS06}
M.~Baroni, C.~Semple, and M.~Steel.
\newblock Hybrids in real time.
\newblock {\em Systematic Biology}, 55:46--56, 2006.

\bibitem{baste2016efficient}
J.~Baste, C.~Paul, I.~Sau, and C.~Scornavacca.
\newblock Efficient {F}{P}{T} algorithms for (strict) compatibility of unrooted
  phylogenetic trees.
\newblock {\em Bulletin of Mathematical Biology}, 79(4):920--938, 2017.

\bibitem{Blair1993}
J.~Blair and B.~Peyton.
\newblock {\em Graph Theory and Sparse Matrix Computation}, chapter An
  Introduction to Chordal Graphs and Clique Trees, pages 1--29.
\newblock Springer New York, New York, NY, 1993.

\bibitem{bodlaender1994tourist}
H.~Bodlaender.
\newblock A tourist guide through treewidth.
\newblock {\em Acta cybernetica}, 11(1-2):1, 1994.

\bibitem{Bodlaender96}
H.~Bodlaender.
\newblock A linear-time algorithm for finding tree-decompositions of small
  treewidth.
\newblock {\em SIAM Journal of Computing}, 25:1305--1317, 1996.

\bibitem{Bodlaender2012}
H.~Bodlaender, F.~Fomin, A.~Koster, D.~Kratsch, and D.~Thilikos.
\newblock On exact algorithms for treewidth.
\newblock {\em ACM Transactions on Algorithms}, 9(1):12:1--12:23, December
  2012.

\bibitem{bodlaender2010treewidth}
H.~Bodlaender and A.~Koster.
\newblock Treewidth computations {I}. upper bounds.
\newblock {\em Information and Computation}, 208(3):259--275, 2010.

\bibitem{bodlaender2011treewidth}
H.~Bodlaender and A.~Koster.
\newblock Treewidth computations {II}. lower bounds.
\newblock {\em Information and Computation}, 209(7):1103--1119, 2011.

\bibitem{bodlaender1997treewidth}
Hans~L Bodlaender and Dimitrios~M Thilikos.
\newblock Treewidth for graphs with small chordality.
\newblock {\em Discrete Applied Mathematics}, 79(1-3):45--61, 1997.

\bibitem{bordewich2017fixed}
M.~Bordewich, C.~Scornavacca, N.~Tokac, and M.~Weller.
\newblock On the fixed parameter tractability of agreement-based phylogenetic
  distances.
\newblock {\em Journal of Mathematical Biology}, 74(1-2):239--257, 2017.

\bibitem{sempbordfpt2007}
M.~Bordewich and C.~Semple.
\newblock Computing the hybridization number of two phylogenetic trees is
  fixed-parameter tractable.
\newblock {\em IEEE/ACM Transactions on Computational Biology and
  Bioinformatics}, 4(3):458--466, 2007.

\bibitem{bryant2006compatibility}
D.~Bryant and J.~Lagergren.
\newblock Compatibility of unrooted phylogenetic trees is {FPT}.
\newblock {\em Theoretical Computer Science}, 351(3):296--302, 2006.

\bibitem{chuzhoy2015excluded}
J.~Chuzhoy.
\newblock Excluded grid theorem: Improved and simplified.
\newblock In {\em Proceedings of the Forty-Seventh Annual ACM on Symposium on
  Theory of Computing (STOC 2015)}, pages 645--654. ACM, 2015.

\bibitem{Courcelle90}
B.~Courcelle.
\newblock The monadic second-order logic of graphs. {I}. {R}ecognizable sets of
  finite graphs.
\newblock {\em Information and Computation}, 85:12--75, 1990.

\bibitem{Cygan:2015:PA:2815661}
M.~Cygan, F.~Fomin, L.~Kowalik, D.~Lokshtanov, D.~Marx, M.~Pilipczuk,
  M.~Pilipczuk, and S.~Saurabh.
\newblock {\em Parameterized Algorithms}.
\newblock Springer Publishing Company, Incorporated, 1st edition, 2015.

\bibitem{diestel2010}
R.~Diestel.
\newblock {\em Graph Theory}.
\newblock Springer-Verlag Berlin and Heidelberg GmbH \& Company KG, 2010.

\bibitem{ding2011agreement}
Yang Ding, Stefan Gr{\"u}newald, and Peter~J Humphries.
\newblock On agreement forests.
\newblock {\em Journal of Combinatorial Theory, Series A}, 118(7):2059--2065,
  2011.

\bibitem{downey2013fundamentals}
R.~Downey and M.~Fellows.
\newblock {\em Fundamentals of parameterized complexity}, volume~4.
\newblock Springer, 2013.

\bibitem{DBLP:conf/gd/DujmovicEW15}
V.~Dujmovic, D.~Eppstein, and D.~Wood.
\newblock Genus, treewidth, and local crossing number.
\newblock In Emilio~Di Giacomo and Anna Lubiw, editors, {\em Graph Drawing and
  Network Visualization - 23rd International Symposium, {GD} 2015, Los Angeles,
  CA, USA, September 24-26, 2015, Revised Selected Papers}, volume 9411 of {\em
  Lecture Notes in Computer Science}, pages 87--98. Springer, 2015.

\bibitem{fischer2014maximum}
M.~Fischer and S.~Kelk.
\newblock On the {M}aximum {P}arsimony distance between phylogenetic trees.
\newblock {\em Annals of Combinatorics}, 20(1):87--113, 2016.

\bibitem{GBP2012}
P.~Gambette, V.~Berry, and C.~Paul.
\newblock Quartets and unrooted phylogenetic networks.
\newblock {\em Journal of Bioinformatics and Computational Biology},
  10(4):1250004, 2012.

\bibitem{Gogate:2004:CAA:1036843.1036868}
V.~Gogate and R.~Dechter.
\newblock A complete anytime algorithm for treewidth.
\newblock In {\em Proceedings of the 20th Conference on Uncertainty in
  Artificial Intelligence}.
\newblock Available online: \texttt{http://graphmod.ics.uci.edu/group/quickbb}.

\bibitem{grigoriev2015}
A.~Grigoriev, S.~Kelk, and N.~Leki\'{c}.
\newblock On low treewidth graphs and supertrees.
\newblock {\em Journal of Graph Algorithms and Applications}, 19(1):325--343,
  2016.

\bibitem{DBLP:journals/jct/GroheM09}
M.~Grohe and D.~Marx.
\newblock On tree width, bramble size, and expansion.
\newblock {\em Journal of Combinatorial Theory, Series {B}}, 99(1):218--228,
  2009.

\bibitem{gysel2012reducing}
R.~Gysel, K.~Stevens, and D.~Gusfield.
\newblock Reducing problems in unrooted tree compatibility to restricted
  triangulations of intersection graphs.
\newblock In Ben Raphael and Jijun Tang, editors, {\em Algorithms in
  Bioinformatics (Proceedings of WABI2012)}, volume 7534 of {\em Lecture Notes
  in Computer Science}, pages 93--105. Springer Berlin Heidelberg, 2012.

\bibitem{humphries2013neighborhoods}
Peter~J Humphries and Taoyang Wu.
\newblock On the neighborhoods of trees.
\newblock {\em IEEE/ACM Transactions on Computational Biology and
  Bioinformatics}, 10(3):721--728, 2013.

\bibitem{HusonRuppScornavacca10}
D.~Huson, R.~Rupp, and C.~Scornavacca.
\newblock {\em Phylogenetic Networks: Concepts, Algorithms and Applications}.
\newblock Cambridge University Press, 2011.

\bibitem{kelk2015reduction}
S.~Kelk, M.~Fischer, V.~Moulton, and T.~Wu.
\newblock Reduction rules for the maximum parsimony distance on phylogenetic
  trees.
\newblock {\em Theoretical Computer Science}, 646 (20):1–15, 2016.

\bibitem{kelk2015}
S.~Kelk, L.~van Iersel, C.~Scornavacca, and M.~Weller.
\newblock Phylogenetic incongruence through the lens of monadic second order
  logic.
\newblock {\em Journal of Graph Algorithms and Applications}, 20(2):189--215,
  2016.

\bibitem{DBLP:journals/jct/Lagergren98}
J.~Lagergren.
\newblock Upper bounds on the size of obstructions and intertwines.
\newblock {\em Journal of Combinatorial Theory, Series {B}}, 73(1):7--40, 1998.

\bibitem{linz2011cluster}
S.~Linz and C.~Semple.
\newblock A cluster reduction for computing the subtree distance between
  phylogenies.
\newblock {\em Annals of Combinatorics}, 15(3):465--484, 2011.

\bibitem{moulton2015parsimony}
V.~Moulton and T.~Wu.
\newblock A parsimony-based metric for phylogenetic trees.
\newblock {\em Advances in Applied Mathematics}, 66:22--45, 2015.

\bibitem{Semple2007}
C.~Semple.
\newblock {\em Reconstructing Evolution - New Mathematical and Computational
  Advances}, chapter Hybridization Networks.
\newblock Oxford University Press, 2007.

\bibitem{SempleSteel2003}
C.~Semple and M.~Steel.
\newblock {\em Phylogenetics}.
\newblock Oxford University Press, 2003.

\bibitem{steel2016phylogeny}
Mike Steel.
\newblock {\em Phylogeny: Discrete and random processes in evolution}.
\newblock SIAM, 2016.

\bibitem{vakati2011graph}
S.~Vakati and D.~Fern{\'a}ndez-Baca.
\newblock Graph triangulations and the compatibility of unrooted phylogenetic
  trees.
\newblock {\em Applied Mathematics Letters}, 24(5):719--723, 2011.

\bibitem{Vakati2015337}
S.~Vakati and D.~Fern{\'a}ndez-Baca.
\newblock Compatibility, incompatibility, tree-width, and forbidden
  phylogenetic minors.
\newblock {\em Electronic Notes in Discrete Mathematics}, 50:337 -- 342, 2015.
\newblock LAGOS'15 – \{VIII\} Latin-American Algorithms, Graphs and
  Optimization Symposium.

\bibitem{vanIersel20161075}
L.~van Iersel, S.~Kelk, and C.~Scornavacca.
\newblock Kernelizations for the hybridization number problem on multiple
  nonbinary trees.
\newblock {\em Journal of Computer and System Sciences}, 82(6):1075 -- 1089,
  2016.

\bibitem{van2016unrooted}
L.~van Iersel, S.~Kelk, G.~Stamoulis, L.~Stougie, and O.~Boes.
\newblock On unrooted and root-uncertain variants of several well-known
  phylogenetic network problems.
\newblock {\em arXiv preprint arXiv:1609.00544}, 2016.

\bibitem{ierselLinz2013}
L.~van Iersel and S.~Linz.
\newblock A quadratic kernel for computing the hybridization number of multiple
  trees.
\newblock {\em Information Processing Letters}, 113(9):318 -- 323, 2013.

\bibitem{ISS2010b}
L.~van van Iersel, C.~Semple, and M.~Steel.
\newblock Locating a tree in a phylogenetic network.
\newblock {\em Information Processing Letters}, 110(23):1037--1043, 2010.

\bibitem{whidden2015calculating}
C.~Whidden and F.~Matsen.
\newblock Calculating the unrooted subtree prune-and-regraft distance.
\newblock {\em arXiv preprint arXiv:1511.07529}, 2015.

\end{thebibliography}

\appendix
\section{Unrooted tree compatibility is FPT: an alternative proof via Courcelle's Theorem}
\label{sec:courcelle}

The Unrooted Tree Compatibility problem (UTC) is simply the problem of determining
whether an unrooted binary phylogenetic network $N$ on $X$ displays an unrooted binary
phylogenetic tree $T$, also on $X$.

In \cite{van2016unrooted} a linear kernel is described for the UTC problem and, separately, a bounded-search branching algorithm. Summarizing, these yield FPT algorithms parameterized by $r(N)$ i.e. algorithms that can solve UTC in time at most $f( r(N) ) \cdot \text{poly}( |N| + |T| )$ for some function $f$ that depends only on $r(N)$. Here we give an
alternative proof of FPT using Courcelle's Theorem. This leverages
Corollary \ref{cor:twbound} from Section \ref{sec:treeNets}. We prove that UTC is not only
FPT when parameterized by $r(N)$, but also $tw(N)$. We begin with a simple
auxiliary observation.

\begin{obs}
Let $G=(V,E)$ be an undirected, connected graph and let $k = |E|-(|V|-1)$. For every $E' \subseteq E$ such that $|E'| = k$,  if $G'=(V, E \setminus E')$ is connected then
$G'=(V,E \setminus E')$ is a spanning tree. \label{obs:connect}
\end{obs}
\begin{proof}
Observe that $G'=(V, E \setminus E')$ has exactly $|V|-1$ edges. If $G'$ contained at least one cycle then (due to the assumption that $G'$ is connected)  we could delete edges
from $G'$, whilst mainting connectivity, to obtain a tree that has strictly fewer than $|V|-1$ edges and which spans $V$. However, this would contradict the standard result
from graph theory that any connected spanning subgraph must contain at least $|V|-1$ edges. Hence, $G'$ is connected and acyclic i.e. it is a spanning tree.
\end{proof}

The high-level idea of the following MSOL formulation is that, if $N$ displays $T$, then (rather like the proof of Lemma \ref{lem:reticbound}) $N$ contains some subtree $T'$
that is a subdivision of $T$ and which can be ``grown'' into a spanning tree $T''$ of $N$. Every spanning tree of $N$ can be obtained by the deletion of $|E|-(|V|-1)$ edges
from $N$, and (from Observation \ref{obs:connect}) any subset $E' \subseteq E$ where $|E'|=k$ and whose deletion from $N$ yields a connected graph, must therefore yield a
spanning tree. Note that the set of quartets (unrooted phylogenetic trees on subsets of exactly
4 taxa) displayed by $T''$ is identical to those displayed by $T'$, which is identical to those displayed by $T$. (In other words,
subdivision operations, and pendant subtrees without taxa that possibly hang from $T''$, do not induce any extra quartets.)

For the benefit of readers not familiar with MSOL we now show how various basic auxiliary predicates can be easily constructed and combined to obtain more powerful predicates.
(The article \cite{kelk2015} gives a more comprehensive inroduction to the use of these techniques in phylogenetics). The MSOL sentence will be queried over the display graph
$D(N,T)$ where we let $V$ be the vertex set of $D(N,T)$ and $E$ its edge set. Here $R^{D}$ is the edge-vertex incidence relation on $D(N,T)$. We let $V_T, V_N, E_T, E_N$ denote
those vertices and edges of $D(N,T)$ which belong to $T, N$ respectively (note that $T_V \cap T_N = X$). Alongside  $X, V, E$ all this information is available to the MSOL
formulation via its \emph{structure}.

\begin{itemize}
\item test that $Z$ is equal to the union of two sets $P$ and $Q$:

$P \cup Q = Z := \forall z ( z \in Z \Rightarrow z \in P \vee z \in Q) \wedge \forall z ( z \in P \Rightarrow z \in Z) \wedge\forall z ( z \in Q \Rightarrow z \in Z).$

\item  test that $P \cap Q = \emptyset$:

$\mathrm{NoIntersect}(P,Q) := \forall u \in P( u \not \in Q ).$

\item test that $P \cap Q = \{v\}$:

$\mathrm{Intersect}(P,Q,v) := (v \in P) \wedge (v \in Q) \wedge \forall u \in P( u \in Q \Rightarrow (u  = v) ).$

\item test if the sets $P$ and $Q$ are a bipartition of $Z$:

$\mathrm{Bipartition}(Z, P, Q) := (P \cup Q = Z) \wedge \mathrm{NoIntersect}(P,Q).$

\item test if the elements in $\{x_1, x_2, x_3, x_4\}$ are pairwise different:

$\mathrm{allDiff}( x_1, x_2, x_3, x_4 ) := \bigwedge_{i \neq j \in \{1,2,3,4\}} x_i \neq x_j.$

\item check if the nodes $p$ and $q$ are adjacent:

$\mathrm{adj}(p,q) := \exists e \in E ( R^{D}(e,p) \wedge R^{D}(e,q)).$
\end{itemize}

The predicate $PAC(Z, x_1, x_2, K)$  asks: is there a path from $x_1$ to $x_2$ entirely contained inside vertices $Z$ that avoids all the edges $K$? We model this by observing
that this does \emph{not} hold if you can partition $Z$ into two pieces $P$ and $Q$, with $x_1 \in P$ and $x_2 \in Q$, such that the only edges that cross the induced cut (if
any) are in $K$.

\begin{eqnarray*}
PAC(Z, x_1, x_2, K) &:=& (x_1 = x_2) \vee \neg \exists P, Q ( \mathrm{Bipartition}(Z,P,Q) \wedge x_1 \in P \wedge x_2 \in Q \wedge\\
&&(\forall p, q (p \in P \wedge q \in Q \Rightarrow \neg \mathrm{adj}( p,q) \vee (\exists g \in K( R^{D}(g, p)\\
&& \wedge R^{D}(g,q))))))
\end{eqnarray*}

The following predicate $QAC^{i}$, where $i \in \{T,N\}$, returns true if and only if $i$ contains an embedding of quartet $x_a x_b | x_c x_d$ that is disjoint from the edge cuts $K$.

\begin{eqnarray*}
QAC^{i}(x_a, x_b, x_c,x_d, K) &:=&
\exists u, v \in V_i ( (u \neq v) \wedge \exists A,B,C,D,P \subseteq V_i ( x_a, u \in A \wedge x_b, u \in B \wedge x_c,\\
&& v \in C \wedge x_d,v \in D \wedge u \in P \wedge v \in P \wedge \mathrm{Intersect}(A,B,u) \wedge\\
&& \mathrm{Intersect}(A,P,u) \wedge \mathrm{Intersect}(B,P,u) \wedge \mathrm{Intersect}(C,D,v) \wedge \\
&&  \mathrm{Intersect}(C,P,v) \wedge \mathrm{Intersect}(D,P,v) \wedge   \mathrm{NoIntersect}(A,C) \wedge \\
&&  \mathrm{NoIntersect}(B,C)\wedge \mathrm{NoIntersect}(A,D)\wedge \mathrm{NoIntersect}(B,D)   \wedge \\
&& PAC(A, u, x_a, K) \wedge PAC(B, u, x_b, K) \wedge PAC(C, v, x_c, K) \wedge \\
&& PAC(D, v, x_d, K)  \wedge PAC(P, u, v, K)))
\end{eqnarray*}

The overall formulation is shown as below. The first line asks for a subset $E'$ of cardinality exactly $k$, the second line requires that the $N$ part of $D(N,T)$ remains
connected after deletion of $E'$ (and thus induces a spanning tree), and from the third line onwards we stipulate that, after deletion of $E'$, the set of quartets that survive
is exactly the same as the set of quartets displayed by $T$. (This is leveraging the well-known
result from phylogenetics that two trees are compatible if and only if they display the same
set of quartets \cite{SempleSteel2003}).

\begin{align*}
|E'| = k \wedge E' \subseteq E_N \wedge \\
 \forall u, v \in V_N(PAC(V_N, u, v, E')) \wedge \\
\forall x_1, x_2, x_3, x_4 \in X( \mathrm{allDiff}(x_1, x_2, x_3, x_4) \Rightarrow \\
((QAC^{T}(x_1, x_2, x_3, x_4, \emptyset) \Leftrightarrow QAC^{N}(x_1, x_2, x_3, x_4, E') )\wedge \\
 (QAC^{T}(x_1, x_3, x_2, x_4, \emptyset) \Leftrightarrow QAC^{N}(x_1, x_3, x_2, x_4, E') ) \wedge \\
 (QAC^{T}(x_1, x_4, x_2, x_3, \emptyset) \Leftrightarrow QAC^{N}(x_1, x_4, x_2, x_3, E') ))).\\
\end{align*}

\begin{thm}
Given an unrooted binary network $N = (V,E)$ and an unrooted binary tree both on $X$, we can determine in time $O( f(k) \cdot n )$ whether $N$ displays $T$, where $k =
|E|-(|V|-1)$ and $n=|V|$.
\end{thm}
\begin{proof}
We run Bodlaender's linear-time FPT algorithm \cite{Bodlaender96} to compute a tree decomposition of $D(N,T)$ and return NO if the treewidth is larger than $k+2$. (This is
correct by Lemma \ref{lem:reticbound}). Otherwise, we construct the constant-length MSOL sentence described earlier and apply the Arnborg et al. \cite{Arnborg91} variant of
Courcelle's Theorem, from which the result follows. (Note that $D(N,T)$ has $O(n)$ vertices and $O(n)$ edges). The result can be made constructive if desired i.e. in the event
of a YES answer the actual set $E'$ can be obtained.
\end{proof}

\begin{cor}
\label{cor:twUTC}
Given an unrooted binary phylogenetic network $N = (V,E)$ and an unrooted binary phylogenetic tree both on $X$, we can determine in time $O( f(t) \cdot n )$ whether $N$ displays $T$, where $t$ is $tw(N)$ and $n=|V|$.
\end{cor}
\begin{proof}
By Corollary \ref{cor:twbound} we know that, if $N$ displays $T$, then $D(N,T)$ will
have treewidth at most $2tw(N)+ 1$. Hence, if the treewidth of $D(N,T)$ is larger
than this, we can immediately answer NO. Otherwise, we have a bound on the
treewidth of $D(N,T)$ in terms of $tw(N)$, so we can again leverage Courcelle's Theorem.
\end{proof}

Corollary \ref{cor:twUTC} is potentially interesting in cases where $tw(N)$ is significantly
smaller than $r(N)$. Note also that $tw(N)$  is a lower bound on the level of $N$ (discussed
earlier in the manuscript: see Section \ref{sec:treeNets}), so in fact UTC is also FPT when parameterized by the level of $N$.

\end{document}